\documentclass[notitlepage, 11pt, tightenlines, aps, pra, superscriptaddress ]{revtex4-1}

\usepackage{textcomp}
\usepackage{gensymb}
\usepackage{amsmath}
\usepackage{amssymb}
\usepackage{bm}
\usepackage{color}
\usepackage[linkcolor=blue, colorlinks=true, urlcolor=blue, citecolor=red]{hyperref}
\usepackage{graphicx}
\usepackage[caption=false]{subfig}
\usepackage{bbold}
\usepackage{verbatim}
\usepackage{appendix}
\usepackage{amsthm}
\usepackage[normalem]{ulem}
\usepackage{tikz}

\usetikzlibrary{calc,decorations.pathmorphing,shapes}

\theoremstyle{plain}
\newtheorem{thm}{Theorem}
\newtheorem{prop}{Proposition}
\newtheorem{lem}{Lemma}
\newtheorem{cor}{Corollary}

\newcommand{\ket}[1]{\left| #1 \right\rangle} 
\newcommand{\bra}[1]{\left\langle #1 \right|} 
\newcommand{\ketbra}[2]{\ket{#1}\hspace*{-0.mm}\bra{#2}}
\newcommand{\ketbras}[3]{\ket{#1}\hspace*{-0.mm}\bra{#2}_{#3}}

\newcommand{\Id}{\mathbb{1}}
\newcommand{\Tr}[1]{\mathrm{\text{Tr}}\left[#1\right]}
\newcommand{\TrS}[2]{\mathrm{\text{Tr}}_{#1}\left[#2\right]}
\newcommand{\Eq}[1]{Eq.~\eqref{#1}}
\newcommand{\bk}{\textbf{\textit{k}}}
\newcommand{\av}[1]{\left\langle #1 \right\rangle} 

\begin{document}
	
\title{Versatile relative entropy bounds for quantum networks}

\author{Luca Rigovacca}
\affiliation{NTT Basic Research Laboratories, NTT Corporation, 3-1 Morinosato-Wakamiya, Atsugi 243-0198, Japan}
\affiliation{Blackett Laboratory, Imperial College London, London SW7 2AZ, United Kingdom}

\author{Go Kato}
\affiliation{NTT Communication Science Laboratories, NTT Corporation, 3-1 Morinosato-Wakamiya, Atsugi 243-0198, Japan}
\affiliation{NTT Research Center for Theoretical Quantum Physics, NTT Corporation, 3-1 Morinosato-Wakamiya, Atsugi 243-0198, Japan}	

\author{Stefan B\"{a}uml}
\affiliation{NTT Basic Research Laboratories, NTT Corporation, 3-1 Morinosato-Wakamiya, Atsugi 243-0198, Japan}
\affiliation{NTT Research Center for Theoretical Quantum Physics, NTT Corporation, 3-1 Morinosato-Wakamiya, Atsugi 243-0198, Japan}	

\author{M. S. Kim}
\affiliation{Blackett Laboratory, Imperial College London, London SW7 2AZ, United Kingdom}

\author{W. J. Munro}
\affiliation{NTT Basic Research Laboratories, NTT Corporation, 3-1 Morinosato-Wakamiya, Atsugi 243-0198, Japan}
\affiliation{National Institute of Informatics, 2-1-2 Hitotsubashi, Chiyoda-ku, Tokyo 101-8430, Japan}
\affiliation{NTT Research Center for Theoretical Quantum Physics, NTT Corporation, 3-1 Morinosato-Wakamiya, Atsugi 243-0198, Japan}

\author{Koji Azuma}
\affiliation{NTT Basic Research Laboratories, NTT Corporation, 3-1 Morinosato-Wakamiya, Atsugi 243-0198, Japan}	
\affiliation{NTT Research Center for Theoretical Quantum Physics, NTT Corporation, 3-1 Morinosato-Wakamiya, Atsugi 243-0198, Japan}

\begin{abstract}
	We provide a versatile upper bound on the number of maximally entangled qubits, or private bits, shared by two parties via a generic adaptive communication protocol over a quantum network when the use of classical communication is not restricted. Although our result follows the idea of Azuma {\it et al.} [Nat. Comm. {\bf 7}, 13523 (2016)] of splitting the network into two parts, our approach relaxes their strong restriction, consisting of the use of a single entanglement measure in the quantification of the maximum amount of entanglement generated by the channels. In particular, in our bound the measure can be chosen on a channel-by-channel basis, in order to make it as tight as possible. This enables us to apply the relative entropy of entanglement, which often gives a state-of-the-art upper bound, on every Choi-simulable channel in the network, even when the other channels do not satisfy this property. We also develop tools to compute, or bound, the max-relative entropy of entanglement for channels that are invariant under phase rotations. In particular, we present an analytical formula for the max-relative entropy of entanglement of the qubit amplitude damping channel.
	\end{abstract}
	
\maketitle

\section{Introduction}
Whenever two parties, say Alice and Bob, want to communicate by using a quantum channel, its noise unavoidably limits their communication efficiency \cite{Wilde_book}. In the limit of many channel uses, their asymptotic optimal performance can be quantified by the channel capacity, which represents the supremum of the number of qubits/bits that can be faithfully transmitted per channel use.
Obtaining an exact expression for this quantity is typically far from trivial. Indeed, in addition to the difficulty of studying the asymptotic behaviour of the channel, the value of the capacity also depends on the task Alice and Bob want to perform, as well as on the free resources available to them \cite{Wilde_book}. 
Two representative tasks, which will be considered in our paper, involve the generation and distribution of a string of shared private bits (pbits) \cite{Horodecki_05,Horodecki_09} or of maximally entangled states (ebits) \cite{Horodecki_ent}. These are known to be fundamental resources for more complex protocols, such as secure classsical communication \cite{Ekert_91,Bennett_92}, quantum teleportation \cite{Bennett_93_teleportation}, and quantum state merging \cite{Horodecki_05_stateMerging}. 
An example of free resource involves the possibility of exchanging classical information over a public classical channel, such as a telephone line or over the internet. Depending on the restrictions on this, the capacity is said to be assisted by zero, forward, backward, or two-way classical communication \cite{Wilde_book}. In this paper we will focus on the last option, that is, no restriction will be imposed on the use of classical communication.

Although the capacity of a quantum channel is by definition an abstract and theoretical quantity, it is also practically useful in that it can be compared with the performance of known transmission schemes. This comparison could then give an indication on the extent of improvements that could be expected in the future.
From this perspective, similar conclusions could be obtained even by studying upper bounds on the channel capacity itself, if they are close enough to its value. For example, with this approach Takeoka {\it et al.} \cite{TGW_nature_14} provided strong evidences for the need of quantum repeaters for long-distance quantum key distribution (QKD) \cite{Munro_2012, Azuma_2015, Azuma_2015b}. 
This reason, together with the fundamental appeal of characterising the ultimate transmission rate achievable by a channel, led to recent intensive research for computable and simple upper bounds on channel capacities, preferably determined by a single use of the channel \cite{TGW_nature_14,TGW_IEEE_14,Pirandola_15,Christandl_16,Goodenough_16,Wilde_2017,Kaur_2017,Cope_2017,Laurenza_2017}.

The results in this direction have been obtained by considering the maximum entanglement that could be shared through a \emph{single} use of a channel $\mathcal N_{A \to B}$, which takes as input a quantum state on Alice's side and yields an output on Bob's one. Indeed, for any entanglement measure $E$ across the bipartition $A:B$, we can define the entanglement of the channel as
\begin{equation}\label{def: E entanglement of channel}
E(\mathcal N) \equiv \max_{\rho_{A A^\prime}} E (\mathcal N_{A^\prime \to B}[\rho_{A A^\prime}]),
\end{equation}
along the lines of Refs.~\cite{TGW_IEEE_14,Christandl_16,Wilde_2017,Pirandola_15}.
For some choices of $\mathcal N_{A\to B}$ and $E$, this can be used to upper bound the private capacity $K(\mathcal N)$, assisted by two-way classical communication. Hence, $E(\mathcal N)$ also acts as an upper bound on the two-way quantum capacity $Q(\mathcal N)$ of the channel, because $Q(\mathcal N) \leq K(\mathcal N)$ (since an ebit can be considered a special case of pbit \cite{Horodecki_05,Horodecki_09}). By generically labeling with $C(\mathcal N)$ one of these two capacities, these upper bounds can be compactly written as
\begin{equation} \label{def: capacity UB}
C(\mathcal N) \leq E(\mathcal N).
\end{equation}
A result of this form has been proven in \cite{TGW_nature_14,TGW_IEEE_14} for \emph{any} quantum channel by employing a particular entanglement measure, the squashed entanglement $E_{\text{sq}}$ \cite{Christandl_04}.  
However, due to the difficulty of computing $E_{\rm{sq}}(\mathcal N)$ exactly \cite{Huang_Ref_2014,Pirandola_15,Goodenough_16}, one often needs to resort to upper bounds on it, thus loosening the bound for the capacity. 
The relative entropy of entanglement $E_{\text{R}}$ is also known to provide an upper bound on the capacity of Choi-simulable quantum channels \cite{Pirandola_15, Wilde_2017}, i.e., channels that can be simulated by performing LOCCs on their Choi-Jamio\l{}kowski states. Quantum channels with this property are also called Choi-stretchable channels \cite{Pirandola_15}.
Remarkably, this upper bound often has no gap with respect to the best known lower bound on the capacity, and when this happens a single-letter formula for the capacity has been found. However, a drawback of the upper bound based on the relative entropy of entanglement is that at the moment it is not known whether \Eq{def: capacity UB}, with $E = E_{\text{R}}$, is valid when applied on a generic, non Choi-simulable, quantum channel. Another option is to use in \Eq{def: capacity UB} the max-relative entropy of entanglement $E_{\text{max}}$ \cite{Christandl_16}.  The resulting bound is formally proven only for quantum channels acting on finite dimensional systems, but it is thought to hold in general (see Ref.~\cite{Christandl_16} for a short discussion).
The set of pairs $(E,\mathcal N)$ for which \Eq{def: capacity UB} is known to hold is the subject of ongoing research, and its extension represents an interesting and challenging problem. 
   
In the future, it is reasonable to expect that all the parties involved in a communication task will be located at different nodes of a quantum network. In this vision, multiple users will be interconnected by a network of quantum channels, which can be utilised with the aim of transmitting or sharing quantum information. This scenario represents the evolution of today's internet in a quantum regime, and is therefore known as ``quantum internet'' \cite{Kimble_08,Azuma_16, Munro_repeaters, Schoute_2016}. Experimental demonstrations of quantum key distribution over metropolitan networks are currently under way \cite{Sangouard_2011, Peev_2009, Stucki_2011, Sasaki,Scarani_QKD_2009}.
Similarly to the single-channel scenario, it is of fundamental and practical importance to seek upper bounds on the rate at which ebits (or pbits) can be shared by two parties by using the channels of the network. This issue has been addressed in Refs. \cite{Pirandola_Network_2016} and \cite{Azuma_16}, where the authors obtained network versions of \Eq{def: capacity UB}, by respectively using $E_\text{R}$ or $E_{\text{sq}}$ as measures of entanglement.
The possibility of dealing with quantum broadcast channels \cite{QBC_def} has also been considered in Refs.~\cite{Yard_11,Seshadreesan_16_broadcast,Baeuml_16,Laurenza_broadcast_2016,Takeoka_broadcast_2016,Takeoka_broadcast_17}. When multiple channels are involved, a typical approach consists in splitting the whole network into two parts, and then in using the maximum amount of entanglement generated by the channels connecting them in order to bound the number of ebits (pbits) produced by a communication protocol.
Thanks to the broad applicability of the single-channel bound given in \Eq{def: capacity UB} for $E = E_{\text{sq}}$, the result of Ref.~\cite{Azuma_16} holds for arbitrary quantum networks.
However, a non-vanishing gap with the optimal number of ebits (or pbits) generated by the network could exist, in analogy with the single-channel case where the upper bounds on the capacity based on the squashed entanglement are typically not tight.
It is thus natural to wonder whether different entanglement measures could improve this sort of network bound, and to what extent the choice of entanglement measure could be tailored to the characteristics of the channels in the network.

In this paper, we start by emphasising how a common strategy is adopted in all the known proofs of the bounds with the form given in \Eq{def: capacity UB}. This allows us to formally identify two sufficient properties that, if satisfied by a given pair $(E,\mathcal N)$, lead to a new instance of \Eq{def: capacity UB}. 
We then show that those two properties also allow us to generalise the result of Ref.~\cite{Azuma_16} on quantum networks to different entanglement measures: $E_{\rm{R}}$ when the channels in the network are Choi-simulable, or $E_{\rm max}$.
The first case is particularly interesting, because \Eq{def: capacity UB} 
is often known to be tighter when stated in terms of $E_{\text{R}}$, rather than in terms of $E_{\text{sq}}$.
The same advantage is therefore expected to be inherited by the corresponding upper bounds on the performance of quantum networks.
However, notice that the $E_{\rm R}$-based bound cannot be applied to arbitrary quantum networks. For example, even if a quantum network is composed almost entirely by Choi-simulable channels that are well bounded by their relative entropy entanglement, the presence of a single channel that is not Choi-simulable forces the use of a weaker entanglement measure (such as $E_{\text{sq}}$) for the whole network.
This suggests that a better bound could be obtained if there was the possibility of changing entanglement measures on a channel-by-channel basis. Our second and most important result goes exactly in this direction.
We exploit an intermediate step in the discussion by Christandl and collaborators in Ref.~\cite{Christandl_16} in order to bound the performance of a quantum network by means of either $E_{\text{R}}$ or $E_{\text{max}}$. 
In particular, as $E_{\text{max}}$ is always larger than $E_{\text{R}}$, we use the relative entropy of entanglement on the Choi-simulable channels of the network, and the max-relative entropy of entanglement on the others. The resulting bound allows us to maintain the precision guaranteed by the relative entropy of entanglement, without the need to restrict its applicability to Choi-simulable networks.
After having presented this general result, we will provide examples of networks where our bound yields an advantage over its counterpart based on the squashed entanglement. In order to do this, we will also evaluate the max-relative entropy of entanglement for the most common qubit channels, by exploiting their symmetry under phase rotations and a recent semidefinite programming (SDP) formulation of $E_{\text{max}}$ \cite{berta2017amortization}. 
In particular, for the qubit amplitude damping channel we are able to analytically solve the SDP optimisation, thus finding the exact expression for its max-relative entropy of entanglement. This quantity upper bounds the private and quantum capacities of the channel assisted by unlimited classical communication, but is less tight than the best known upper bound based on the squashed entanglement \cite{Pirandola_15}.

The remainder of this paper is organised as follows. In Sec. \ref{sec: preliminaries} we introduce our notation and some preliminary notions that will be used in the following. In Sec. \ref{sec: abastract approach} we formally identify sufficient properties that, if satisfied by a pair $(E,\mathcal N)$, lead to an upper bound on the capacity of the channel as in \Eq{def: capacity UB}. Furthermore, along the lines of Ref.~\cite{Azuma_16}, we show how the same properties are also sufficient to obtain an upper bound on the number of ebits (or pbits) generated through a quantum network. 
Our main result is presented in Sec. \ref{sec: hybrid bound}, where we derive a similar versatile upper bound, in which different entanglement measures are applied to the channels of the network depending on their Choi-simulability.
Analytical or numerical evaluations of the max-relative entropy of entanglement for the most common qubit channels can be found in Sec.~\ref{sec: channels}, while examples of networks where our bound performs better than the one based on the squashed entanglement are presented in Sec.~\ref{sec: examples}. A final discussion on our results can be found in Sec.~\ref{sec: disc and conclusions}, together with our conclusions. Technical details are left for the appendices.

\section{Preliminaries} \label{sec: preliminaries}
In this section we introduce the basic concepts necessary to understand the remainder of the paper, and we describe the notation we will use. In particular, we start by looking at the definitions and properties of the relative and max-relative entropy. Then, we introduce the notion of private states and of Choi-simulable channels. We also formally describe the structure of a quantum network and of the most general adaptive protocol, assisted by free classical communication, that could be employed to share ebits (or pbits). 
At the end of the section, we discuss the figure of merit we use to quantify the performance of a given communication strategy, and we comment on its relation to the usual single-channel capacity.

\subsection{Relative and max-relative entropies}\label{sec: entropy defs}
Given two quantum states $\rho$ and $\sigma$, with supports satisfying $\text{Supp}(\rho) \subseteq \text{Supp}(\sigma)$, their relative entropy \cite{umegaki1962} and max-relative entropy \cite{Datta_09} are respectively defined as
\begin{align}
S(\rho\Vert\sigma) &= \Tr{\rho(\log_2\rho - \log_2\sigma)},\\
D_{\text{max}}(\rho\Vert\sigma) &= \inf \{x \in \mathbb{R}\vert 2^x \sigma - \rho \geq 0 \}, \label{def: Dmax}
\end{align}
while their values are set to $\infty$ if the condition on the supports is not satisfied. The relative and max-relative entropy of two states are related by
\begin{equation} \label{eq: S Dmax ineq}
S(\rho\Vert\sigma) \leq D_{\text{max}}(\rho\Vert\sigma),
\end{equation}
they are also non-negative, equal to zero if and only if $\rho = \sigma$, and invariant under joint unitary operations, that is:
\begin{equation}\label{def: invariance joint unitaries}
S(U \rho U^\dagger \Vert U \sigma U^\dagger) = S(\rho \Vert \sigma), \qquad D_{\text{max}}(U \rho U^\dagger \Vert U \sigma U^\dagger) = D_{\text{max}}(\rho \Vert \sigma),
\end{equation}
for any unitary $U$.
Moreover, the relative entropy is jointly convex in its arguments \cite{JointConvexity}, whereas the max-relative entropy is jointly quasi-convex:
\begin{align}
S\left(\sum_i p_i \rho_i \bigg\Vert \sum_i p_i \sigma_ i\right) &\leq \sum_i p_i S(\rho_i\Vert\sigma_i), \\
D_{\text{max}}\left(\sum_i p_i \rho_i \bigg\Vert \sum_i p_i \sigma_ i\right) &\leq \max_i  D_{\text{max}}(\rho_i\Vert\sigma_i), \label{def: joint quasi convexity}
\end{align}
where $\{\rho_i\}_i$ and $\{\sigma_i\}_i$ are quantum states, and $p_i\geq 0$ with $\sum_i p_i = 1$.

The relative and max-relative entropies can be used to define entanglement measures respectively known as relative entropy of entanglement \cite{Plenio_1998} and max-relative entropy of entanglement \cite{Datta_09}. For a given bipartite state $\rho_{AB}$, their values are obtained by optimising over all separable states as follows:
\begin{align} \label{def: ER}
E_{\text{R}}^{A:B}(\rho_{AB}) & = \min_{\sigma_{AB}\in \text{SEP} } S(\rho_{AB}\Vert\sigma_{AB}),\\
E^{A:B}_{\rm{max}}(\rho_{AB}) & = \min_{\sigma_{AB}\in \text{SEP}} D_{\text{max}}(\rho_{AB}\Vert\sigma_{AB}). \label{def: Emax}
\end{align}
In the following we do not explicitly write the bipartition $A:B$ in the symbols $E_{\text{R}}$ and $E_{\text{max}}$, unless needed to avoid confusion. If the local quantum systems of Alice (or Bob) are divided into smaller subsystems, these will be labelled for example as $A,A^\prime,A^{\prime\prime}$ (or $B,B^\prime,B^{\prime\prime}$). In this case, the default evaluation of an entanglement measure has to be considered across the bipartition $AA^\prime A^{\prime\prime}:BB^\prime B^{\prime\prime}$.
As any good entanglement measure, $E_{\text{R}}$ and $E_{\text{max}}$ are, on average, monotonically non-increasing under local operations and classical communication (LOCC). For an entanglement measure $E$, this property can be explicitly written as
\begin{equation}\label{eq: LOCC monotonicity}
\sum_k p_k E\left(\rho^{(k)}_{AB}\right) \leq E(\rho_{AB}), 
\end{equation}
where $k$ represents the measurement outcome of the LOCC operation applied on $\rho_{AB}$, $p_k$ is the probability of obtaining it, and $\rho^{(k)}_{AB}$ is the output state of the system post-selected on that result.
Moreover, the ordering relation in \Eq{eq: S Dmax ineq} can also be straightforwardly extended to the entanglement measures $E_{\text{R}}$ and $E_{\text{max}}$, as well as to the entanglement of a channel $\mathcal N$ [see \Eq{def: E entanglement of channel}]:
\begin{equation}\label{eq: Er Emax relation}
E_{\text{R}}(\rho_{AB}) \leq E_{\text{max}}(\rho_{AB}), \qquad E_{\text{R}}(\mathcal N) \leq E_{\text{max}}(\mathcal N).
\end{equation}
Further details on $E_{\text{max}}$ can be found in Ref. \cite{Datta_09a}.

We stress that in the remainder of this paper any generic entanglement measure $E$ satisfies \Eq{eq: LOCC monotonicity}, and becomes zero when evaluated on any separable state.

\subsection{Target states: maximally entangled or private states}
The typical goal of two parties, say Alice and Bob, in a quantum communication protocol  is to share one or multiple copies of a $d$-dimensional maximally entangled state
\begin{equation} \label{def: max ent state}
\psi_{AB}(d) = \sum_{i,j=1}^{d} \frac{1}{d}\ketbras{ii}{jj}{AB},
\end{equation}
where $\{\ket{i}_{A (B)}\}_i$ forms a local orthonormal basis. Any single copy of these states corresponds to $\log_2 d$ ebits, which Alice and Bob can use to perform one of many possible tasks. For example, they can transmit any $d$-dimensional state via the teleportation protocol, or they can perform a projective measurement on it in order to share a string of $\log_2 d$ bits of private randomness. The maximally entangled state, however, is not the only quantum state from which a private key can be obtained by performing local measurements. It has been shown that this is possible whenever Alice and Bob are able to distill via LOCC a so-called ``private state'' \cite{Horodecki_05,Horodecki_09}, which has the following form:
\begin{equation}\label{def: private state}
\gamma_{AB A^\prime B^\prime}(d) = U^{(\text{twist})}_{AB A^\prime B^\prime} \left(\psi_{AB}(d)\otimes \sigma_{A^\prime B^\prime}\right) U^{(\text{twist})\dagger}_{AB A^\prime B^\prime}.
\end{equation}
The state $\sigma_{A^\prime B^\prime}$ is arbitrary and the controlled unitary
\begin{equation}
U^{(\text{twist})}_{AB A^\prime B^\prime} = \sum_{ij=1}^d \ketbras{i}{i}{A}\otimes\ketbras{j}{j}{B}\otimes U_{A^\prime B^\prime}^{(ij)}
\end{equation}
is known as ``twisting unitary'', with each $U_{A^\prime B^\prime}^{(ij)}$ a unitary operator.
The local subsystems $A$ and $B$ are called ``key systems'', whereas $A^\prime$ and $B^\prime$ are known as ``shield systems''. The role of the latter is to prevent an eavesdropper from getting access to the key component, and they could have any dimension.

\subsection{Choi-simulable channels}
The idea of using quantum teleportation in order to simplify the structure of a computation for communication task has been used several times in the past \cite{Bennett_1996,Gottesman_99,horodecki_tp_1999,KLM,Wolf_tp_2007,Cerf_tp_2009,Cerf_tp_2009,MH_thesis}.
Recently, a similar idea has been used in Ref.~\cite{Pirandola_15} and in Refs. \cite{Pirandola_Network_2016,Wilde_2017} in order to obtain upper bounds on the capacities of quantum channels $\mathcal N$ such that their action on a quantum state $\tilde\rho_{A^\prime}$ can be written as
	\begin{equation}\label{eq: streatchability condition}
	\mathcal N_{A^\prime \to B^\prime}(\tilde\rho_{A^\prime}) = \Lambda_{A^{\prime} A^{\prime\prime}:B^\prime} \left(\tilde\rho_{A^\prime} \otimes \pi_{A^{\prime\prime}B^\prime}(\mathcal N) \right).
	\end{equation}
Here $\Lambda_{A^{\prime} A^{\prime\prime}:B^\prime}$ is a trace-preserving LOCC operation and $\pi_{A^{\prime\prime}B^\prime}(\mathcal N) = \mathcal N_{\tilde A\to B^\prime} (\psi_{A^{\prime\prime} \tilde A})$ represents the Choi-Jamio\l{}kowski state associated with the quantum channel $\mathcal N$, with $\psi_{A^{\prime\prime} \tilde A}$ a maximally entangled state.
We will say that channels satisfying \Eq{eq: streatchability condition} are Choi-simulable, as they can be simulated by applying LOCCs to their Choi-Jamio\l{}kowski state. This property can also go under the name of ``Choi-stretchability'' \cite{Pirandola_15,Pirandola_Network_2016}. 
The importance of \Eq{eq: streatchability condition} lies in the fact that it gives the possibility of reducing the effect of a quantum channel to the presence of an initially shared Choi state, up to some LOCC transformation. 
Equation \eqref{eq: streatchability condition} makes the description of the quantum communication much simpler, because the LOCC $\Lambda_{A^{\prime} A^{\prime\prime}:B^\prime}$ can be included among those freely performed by the parties.
In the following, if a channel $\mathcal N$ is Choi-simulable we will write $\mathcal N \in \mathcal S$. 

Remarkably, the relative entropy of entanglement of a Choi-simulable channel $\mathcal N$, as defined in \Eq{def: E entanglement of channel}, provides an upper bound on its capacity assisted by two-way classical communication \cite{Pirandola_15, Wilde_2017}.
Moreover, $E_{\text{R}}(\mathcal N)$ exactly coincides with the capacity $C(\mathcal N)$ on a particular subset of Choi-simulable channels, whose capacities $C(\mathcal N)$ can thus be written as  single-letter formulas \cite{Pirandola_15}. Channels for which this happens can also be called ``distillable'' \cite{Pirandola_15,Pirandola_Network_2016}.
Among these, we can enumerate the erasure and dephasing channels in finite dimensional systems, as well as the bosonic lossy channel.
Interestingly, for many Choi-simulable channels (such as Pauli channels) $E_{\text{R}}(\mathcal N)$ turns out \cite{Pirandola_15} to be a tighter upper bound on $C(\mathcal N)$ than other known upper bounds based on the squashed entanglement \cite{TGW_nature_14,TGW_IEEE_14}. 
However, one should keep in mind that this is not always the case, as can be seen by considering a channel having an antisymmetric Choi state. Indeed, the squashed entanglement of this state, and thus of the associated quantum channel, can be arbitrarily small compared to its relative entropy of entanglement \cite{Christandl_antisymm_2010,Christandl_antisymm_2012}.

We now explicitly derive a property that the relative entropy of entanglement satisfies when applied on the output of a Choi-simulable channel. Although it is obvious from the discussion in Ref.~\cite{Pirandola_15}, it is beneficial to go through its proof in detail, because it will play a central role in the remainder of this paper.   In particular, we prove that if $\tilde \rho_{AB^\prime B}$ is obtained as output of a Choi-simulable channel $\mathcal N \in \mathcal S$ as
\begin{equation}
\tilde \rho_{AB^\prime B} = \mathcal N_{A^\prime \to B^\prime}(\rho_{AA^\prime B}),
\end{equation}
the following chain of inequalities holds:
\begin{align}\label{eq: ER bound}
E_{\text{R}}(\tilde\rho_{A B^\prime B}) &\leq E_{\text{R}}\left(\rho_{AA^\prime B} \otimes \pi_{A^{\prime\prime}B^\prime}(\mathcal N)\right)\notag\\
& \leq E_{\text{R}}\left(\pi_{A^{\prime\prime}B^\prime}(\mathcal N)\right) + E_{\text{R}}\left(\rho_{AA^\prime B}\right)\notag\\
& = E_{\text{R}}\left(\mathcal N_{A^\prime \to B^\prime}\right) + E_{\text{R}}\left(\rho_{AA^\prime B}\right).
\end{align}
The first inequality comes from \Eq{eq: streatchability condition} and from the monotonicity of $E_{\text{R}}$ under LOCC, while the second one follows from its sub-additivity under tensor products. The final equality can be proven by showing inequalities in both directions. Indeed, the inequality ``$\leq$'' is obtained by noticing that a maximisation over all input states would be needed in order to obtain the relative entropy of entanglement of a channel [see \Eq{def: capacity UB}]. The converse direction, instead, is once again a consequence of \Eq{eq: streatchability condition} and of the monotonicity of $E_{\text{R}}$ under LOCC \cite{Pirandola_15}:
\begin{align}
E_{\text{R}}(\mathcal N_{A^\prime\to B^\prime}[\rho_{A A^\prime}]) &= E_{\text{R}}\left(\Lambda_{A^{\prime} A^{\prime\prime}:B} \left[\rho_{A A^\prime} \otimes \pi_{A^{\prime\prime}B^\prime}(\mathcal N) \right]\right) \leq E_{\text{R}}\left( \pi_{A^{\prime\prime}B}(\mathcal N) \right),
\end{align}
which holds for any $\rho_{A A^\prime}$ and thus also for its maximum value $E_{\text{R}}\left(\mathcal N_{A^\prime \to B^\prime}\right)$. Hence, \Eq{eq: ER bound} shows that the amount of entanglement which can be found in output of a Choi-simulable channel, as measured by $E_{\text{R}}$, can be upper bounded by the amount already present in input plus the maximum amount that can be created by the channel itself. Up to date, it is not known whether the same conclusion could be obtained also for any quantum channel.

\subsection{Quantum networks as graphs}\label{sec: network and graph}
The simplest setup that allows Alice and Bob to exchange quantum information is shown in Fig.~\ref{fig: single channel}, where a quantum channel $\mathcal{N}_{A\to B}$ connects Alice's laboratory with Bob's. More generally, we can think of them as being two local users having access to a quantum network, as in Fig.~\ref{fig: quantum network}.
A quantum network is composed of several nodes, connected by many quantum channels potentially different from each other. 
We can formally describe this structure by a directed graph $G = (V,L)$, where $V = \{V_0, \ldots, V_{M+1}\}$ is the set of nodes and $L$ is the set of directed edges, or links, between the nodes. For any edge $l = (V_i,V_j) \in L$, there is a quantum channel $\mathcal N^{(l)}$ from node $V_i$ to node $V_j$. 
Without loss of generality, we can assume that nodes $A = V_0$ and $B = V_{M+1}$ are respectively controlled by Alice and Bob, whereas the remaining nodes $\{C_i\}_{i=1}^M$, with $C_i = V_i$, are not.

\begin{figure}
	\centering
	\subfloat[Single channel. \label{fig: single channel}]{\includegraphics[scale=0.5]{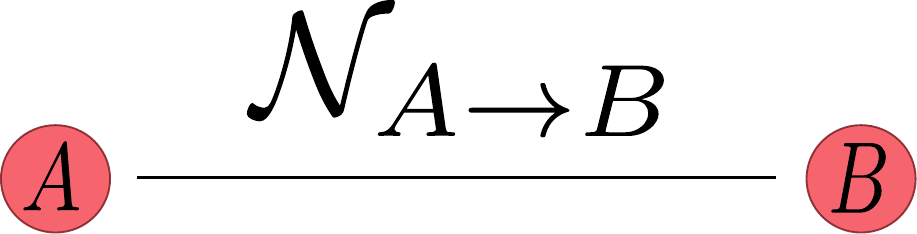}}\\
	\subfloat[Quantum network. \label{fig: quantum network}]{\includegraphics[scale=0.5]{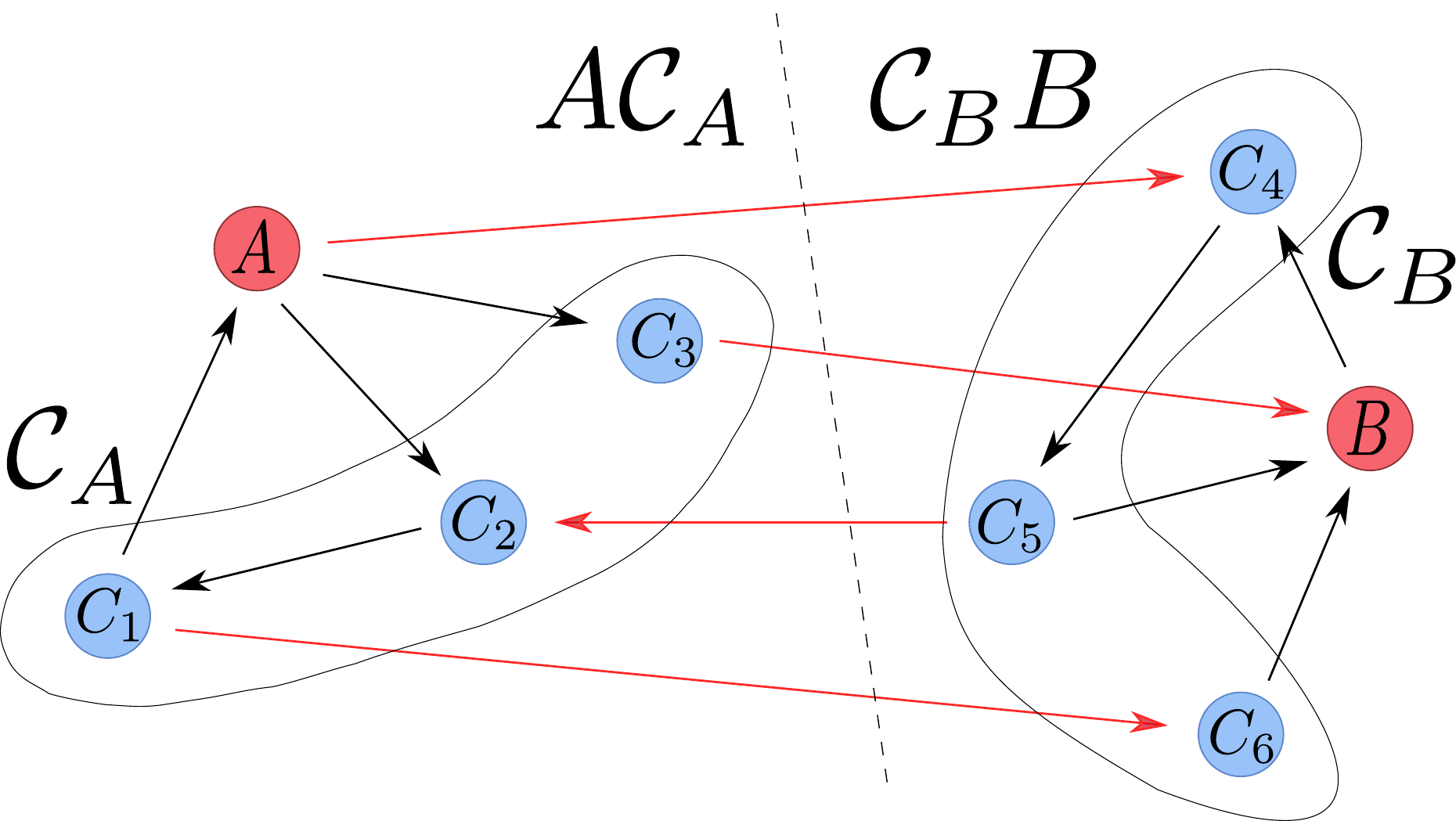}}
	\caption{
		(a) Single-channel communication scenario, where $A$ and $B$ are connected through the channel $\mathcal{N}_{A\to B}$.
		(b) An example of quantum network, with $M=6$ additional nodes. Every arrow corresponds to a quantum channel. The  bipartition $\mathcal C_A = \{C_1,C_2,C_3\}$, $\mathcal C_B = \{C_4,C_5,C_6\}$ is shown as an example, and the channels $\mathcal N^{(l)}$ with $l \in L_{\mathcal C_A}$, which connect the two partitions, are coloured in red. 
	}
\end{figure}

In the following we will often make use of the notion of ``bipartition'' of a quantum network. This is defined by dividing the nodes $\{C_i\}_i$ into two disjoint sets: $\mathcal C_A \subset \{C_i\}_i$ and $\mathcal C_B = \{C_i\}_i \setminus \mathcal C_A$. Once a bipartition has been chosen, the set of edges connecting the nodes in $\{A\}\cup \mathcal C_A$ with those in $\mathcal C_B \cup \{B\}$, or vice versa, will be labelled as $L_{\mathcal C_A} \subset L$. Moreover, in order to keep our notation simple, in the remainder of this paper we will refer to the subsets of nodes $\{A\}\cup \mathcal C_A$ and $\mathcal C_B \cup \{B\}$ by writing respectively $A \mathcal C_A$ and $\mathcal C_B B$.

\subsection{Adaptive strategy over quantum networks} \label{sec: adaptive strategy}

We assume that full quantum control over the local systems is available on each node of the network, and that all parties can freely exchange classical information at any stage of the protocol in order to coordinate their strategy. Moreover, we also assume that every node in the network will collaborate with Alice and Bob in order to allow them to achieve their goal.
At the beginning of the most general adaptive communication protocol, the parties initialise their systems in a separable state $\rho^{(1)}_{ABC_1\ldots C_M}$. Then, they iteratively exchange (part of) their systems via the quantum channels, and perform LOCCs on the obtained states, which may involve measurements. For this reason, every choice made by the parties at a certain stage of the protocol may depend on all previously obtained LOCC outcomes. In the remainder of this section we formally describe any protocol of this kind, similarly to what has been done in Refs. \cite{Azuma_16,Pirandola_Network_2016,Baeuml_16}. For the sake of simplicity, we will drop the subscript $ABC_1\ldots C_M$ from states spread over the whole network.

Between any two channel uses several LOCC may be performed, but we can group them into a single ``round of LOCCs'' yielding an overall multi-index outcome $k$. In this way, a single ``round of the protocol'' will be composed by the application of a channel followed by a round of LOCCs.
Let us group within the vector ${\bk}_i = (k_0,k_1, \ldots, k_{i-1},k_i)$ the sequence of LOCC outcomes obtained in the first $i$ rounds, with $k_0 \equiv 1$ added for convenience. 
In this way, the $i$-th round of the protocol receives as input $\rho^{\bk_{i-1}}$ and transforms it into $\rho^{\bk_{i}}$ via the following two steps.
\begin{itemize}
	\item Depending on the previous LOCC outcomes, grouped within $\bk_{i-1}$, the parties may use the channel $\mathcal N^{(l_{\bk_{i-1}})}$ to transmit a quantum state along the edge $l_{\bk_{i-1}} \in L$ of the graph $G$ characterising the network. The global state at the end of this step is labeled by $\tilde\rho^{\bk_{i-1}}$;
	\item A round of LOCCs $\Lambda^{(\bk_{i-1})}$ is performed on $\tilde\rho^{\bk_{i-1}}$, with output $k_i$ obtained with probability $p(k_i|\bk_{i-1})$. The output quantum state $\rho^{\bk_i}$ will be used as input for the following round of the protocol.
\end{itemize} 
When the protocol stops, say after $n$ rounds, the final state $\rho^{\bk_n}_{AB} = \TrS{C_1,\dots,C_M}{\rho^{\bk_n}}$ shared by Alice and Bob
has to be $\epsilon$-close in trace distance to an ideal target state $\phi_{AB}(d_{\bk_n})$, i.e., such that for any sequence of outcomes $\bk_n$ one has
\begin{equation}\label{eq: epsilon close}
\Vert\rho^{\bk_n}_{AB}-\phi_{AB}(d_{\bk_n})\Vert_1 = \epsilon,
\end{equation}
where $\Vert O\Vert_1 \equiv \text{Tr}\left[\sqrt{O^\dagger O}\right]$.
The target state $\phi_{AB}(d_{\bk_n})$ can either be a maximally entangled state $\psi_{AB}(d_{\bk_n})$ [see \Eq{def: max ent state}] or a private state $\gamma_{AB}(d_{\bk_n})$ [see \Eq{def: private state}], depending on the task of Alice and Bob. 

All the details of the adaptive strategy leading to \Eq{eq: epsilon close} are determined by the protocol $\mathcal P_{\epsilon,n}$ that the parties are following. These details include the error threshold $\epsilon$, the maximum number of rounds $n$, the target states $\phi_{AB}(d_{\bk_n})$, and the set of rules that, at any round of the protocol, map the vectors of previous outcomes $\{\bk_i\}_{i=0}^{n-1}$ to the channel and LOCC operations used in the following. In the remainder of this paper we will often have to average some function $F(\bk_n)$ over all possible LOCC outcomes $\{\bk_n\}$. It is thus convenient to introduce the shorthand notation
\begin{equation}
\av{F}_{\mathcal P_{\epsilon,n}} \equiv \sum_{\bk_n} p(\bk_n) F(\bk_n),
\end{equation}
where $p(\bk_n)$ is the probability of obtaining this particular sequence of LOCC outcomes according to the protocol $\mathcal P_{\epsilon,n}$.

We point out that the number of channels used in the protocol will generally be smaller than $n$. This is because in any round the parties \emph{may} decide to use a channel of the network, but are not forced to do so.
However, without loss of generality we can assume that a channel is used in any round of the protocol up to a certain point, after which the parties can only perform LOCCs and the communication protocol is effectively aborted.
Indeed, if this were not the case, we could recover this situation simply by merging all the LOCCs performed between two channel uses into a single round of LOCCs. Notice that depending on the LOCC outcomes already obtained, the parties can decide to effectively abort the communication after different numbers of channel uses.
In particular, for any edge $l \in L$ and vector $\bk_n$, we can define as $m^{(l)}(\bk_n)$ the total number of times channel $\mathcal N^{(l)}$ has been used in that particular realisation of outcomes. Formally, this can be written as
\begin{equation}\label{def: m for a channel}
m^{(l)}(\bk_n) = \sum_{i=0}^{n-1} \delta_{l,l_{\bk_i}},
\end{equation}
where the symbol $\delta$ represents the Kronecker delta, while the total number of channel uses is
\begin{equation}
m(\bk_n) = \sum_{l \in L} m^{(l)}(\bk_n).
\end{equation}
A value of $m(\bk_n)$ strictly smaller than $n$ means that the protocol has been effectively interrupted after $m(\bk_n)$ rounds.

\subsection{Quantifying the performance of a communication protocol}
The quality of a point-to-point adaptive communication protocol $\mathcal P_{\epsilon,n}$ can be quantified by its ability to produce a large number of shared ebits (or pbits) between  Alice and Bob. For any realisation $\bk_n$ of LOCC outcomes, this corresponds to the logarithm of the dimension $d_{\bk_n}$ that characterises the target state $\phi_{AB}(d_{\bk_n})$, $\epsilon$-close to the final state $\rho^{\bk_n}_{AB}$ produced by the protocol. Therefore, a good figure of merit for $\mathcal P_{\epsilon,n}$ can be obtained by averaging this quantity over all LOCC outcomes. In our notation, this can be written as $\av{\log_2 d}_{\mathcal P_{\epsilon,n}}$. 

This approach is particularly suitable to characterise the performance of protocols that use the channels of the network a finite number of times, because it directly provides the length of ebits (pbits) that Alice and Bob can expect to share at the end of the communication. 
However, the quantity $\av{\log_2 d}_{\mathcal P_{\epsilon,n}}$ becomes unbounded when the asymptotic limit of infinitely many channel uses is considered.
In a single-channel scenario, this issue has been traditionally addressed by considering the communication rate, i.e., the number of bits produced per channel use. We should point out that in this case one does not typically consider the possibility of interrupting the protocol depending on previous LOCC outcomes. This is because otherwise with non-zero probability the asymptotic regime of infinitely many channel uses would not be reached. For this reason  only protocols which use the quantum channel after \emph{every} round of LOCC are normally considered when assessing its asymptotic performance. In this paper, a protocol of this kind will be labelled as $\tilde{\mathcal P}_{\epsilon,N}$, where $\epsilon$ represents the error threshold and $N$ is the fixed number of channel uses. With this notation, the quantum (or private) capacity of a quantum channel $\mathcal N$ assisted by two-way classical communication can be obtained as the limit
\begin{equation}\label{def: capacity}
C(\mathcal N) = \lim_{\epsilon\to 0} \lim_{N\to \infty} \sup_{\tilde{\mathcal P}_{\epsilon,N}} \frac{\av{\log_2 d}_{\tilde{\mathcal P}_{\epsilon,N}}}{N}.
\end{equation}

For a generic quantum network the situation is more involved, and in the literature one can find multiple ways of assessing its communication performance in the asymptotic limit. For example, one can fix the frequency with which each channel is used, and divide $\av{\log_2 d}_{\mathcal P_{\epsilon,n}}$ by the total number of channel uses \cite{Azuma&Kato_2016}. Other options, proposed in Ref.~\cite{Pirandola_Network_2016}, consist in using each ``path'' connecting Alice and Bob with a certain probability, or in using each channel of the network exactly once. Then, the number of produced ebits (pbits) is respectively divided by the number of paths used, or by the total number of times the network has been accessed.
Although the details of characterising the considered figure of merit can change on a case-by-case basis, one typically has to optimise $\av{\log_2 d}_{\mathcal P_{\epsilon,n}}$ over a chosen class of protocols, and divide it by a quantity that counts how many times a basic operation has been repeated. 
 
Similar to Refs. \cite{Azuma_16,Azuma&Kato_2016}, in the following we are able to provide an upper bound on $\av{\log_2 d}_{\mathcal P_{\epsilon,n}}$ for a generic adaptive protocol running on a quantum network with graph $G$. This bound only depends on the maximum amount of entanglement that could be generated by the quantum channels composing the network, and on the number of times each channel has been used. 
From the previous discussion, it should be clear that our bound can be easily converted to a bound on a broad class of figures of merit, which could be chosen to quantify the performance of the network. For example, in the case of a single channel, our bound can be connected to an upper bound on the capacity by means of \Eq{def: capacity}.

\section{Entanglement-based upper bounds} \label{sec: abastract approach}
As mentioned in the introduction, recently several authors provided bounds on the number of ebits (pbits) shared by two parties at the end of a point-to-point communication protocol assisted by two-way classical communication. Some studies deal with the capacity of a single quantum channel \cite{TGW_nature_14,TGW_IEEE_14,Christandl_16,Pirandola_15,Wilde_2017}, whereas others consider quantum networks with arbitrary topology \cite{Azuma_16,Azuma&Kato_2016,Pirandola_Network_2016}. 
However, they all share some common features. Here we identify these, and show how they can lead to known, new, or yet to be discovered communication bounds.

We start by considering a single channel $\mathcal N$ and a generic entanglement measure $E$, and we formally summarise in Theorem \ref{thm: key hypothesis} some important properties that have been used in the past in order to obtain upper bounds on the channel capacity. One advantage of this abstract formulation is that it can ease the process of identifying all the entanglement measures which can be used to bound the capacity of a given channel. By comparing all these bounds, it would then be possible to select the one with the minimum value, which represents the best known upper bound on the capacity $C(\mathcal N)$.
A second advantage of our abstract approach lies in the possibility of easily extending previous results on quantum networks to other entanglement measures, not explicitly studied in the original papers. This is because the same properties responsible for the upper bound on the capacity of a single channel are also the main ingredients used in Ref.~\cite{Azuma_16} to derive an upper bound for the number of shared ebits (pbits) produced by a quantum network.
In this way, we are able to show that the same bound of Ref.~\cite{Azuma_16}, originally expressed in terms of the squashed entanglement, is also valid for other entanglement measures: $E_{\rm{max}}$ and $E_{\text{R}}$, although the applicability of the latter is restricted to networks composed by Choi-simulable channels. This original contribution will be summarised as Theorem \ref{thm: ent-based network bound}.
 
Having multiple upper bounds on the communication performance of a quantum network, based on different entanglement measures, there is the possibility of combining them together in order to obtain a bound as tight as possible.
An obvious option consists in evaluating each upper bound separately, and then selecting the one which yields the minimum value. However, it is possible to do better than this, and in Sec.~\ref{sec: hybrid bound} we show how the bounds based on $E_{\text{max}}$ and $E_{\text{R}}$ can be joined together to form a single tighter bound.

\subsection{General Framework}
We start by discussing the case of a single channel $\mathcal N$, and then we move to the more general situation of a quantum network with arbitrary topology. The proofs for the theorems presented here can be found at the end of the section. 

All measures of entanglement $E$ known to yield a bound on the number of ebits (pbits) generated by a communication protocol satisfy the following property:
\begin{enumerate}
	\item[P1.] If a target state $\phi_{AB}(d)$ is $\epsilon$-close to a quantum state $\rho_{AB}$, i.e., if $\Vert\rho_{AB} - \phi_{AB}(d)\Vert_1 = \epsilon$,  then there exist two real functions $f_E$ and $g_E$, with $\lim_{\epsilon \to 0} g_E(\epsilon) = 1$ and $\lim_{\epsilon \to 0} f_E(\epsilon) = 0$, such that 
\begin{equation}\label{def: P1}
E(\rho_{AB} ) \geq g_E(\epsilon) \log_2 d - f_E(\epsilon).
\end{equation}
\end{enumerate}
For a maximally entangled target state, this property can be easily proven for every \emph{asymptotically continuous} \cite{SynakRadtke_06} measure $E$. On the contrary, more effort is usually required to prove it for private target states. The reason for this is that the quantity $d$ appearing on the right-hand side of \Eq{def: P1} needs to be the dimension of the key systems, rather than the dimension of the whole key-shield systems.  
Nonetheless, property $P1$ has been proven for $E_{\text{sq}}$ \cite{Wilde_16} and $E_{\text{R}}$ \cite{Horodecki_09,Wilde_2017,Pirandola_15}. It can also be easily proven for $E_{\text{max}}$, by slightly varying the proof of Lemma IV.2 in Ref.~\cite{Christandl_16} in order to obtain \Eq{def: P1} with 
\begin{equation} \label{eq: fg Emax}
g_{E_{\text{max}}} = 1, \qquad \text{and}\qquad f_{E_{\text{max}}} = - 2 \log_2(1-\epsilon/2).
\end{equation}
Another important property of a pair concerns the relation between the amount of entanglement in the input and output states of the channel $\mathcal N$, as measured by the entanglement measure $E$. A pair $(E,\mathcal N)$ is said to satisfy property $\rm{P2}$ if for all states $\rho_{AA^\prime B}$ one has
\begin{enumerate}
\item[P2.] $\tilde \rho_{AB^\prime B} = \mathcal N_{A^\prime \to B^\prime}(\rho_{AA^\prime B}) \quad\Longrightarrow\quad  E(\tilde \rho_{A B^\prime B})\leq E(\mathcal N) + E(\rho_{AA^\prime B})$,
\end{enumerate}
where $E(\mathcal N)$ is the maximum entanglement shared through a single use of the channel [see \Eq{def: E entanglement of channel}].
This is known to hold for any quantum channel when $E = E_{\text{sq}}$ \cite{TGW_nature_14,TGW_IEEE_14}, and for any Choi-simulable channel when $E = E_{\text{R}}$ [see \Eq{eq: ER bound}]. 
Moreover, property $P2$ has been recently shown for the max-relative entropy of entanglement for any channel acting on finite-dimensional system, but it is conjectured to hold even without this assumption \cite{Christandl_16}.

In order to ease the connection with quantum networks, we provide a bound on $\av{\log_2 d}_{\mathcal P_{\epsilon,n}}$ also in the single-channel scenario, from which the usual bound on the capacity can be recovered as a corollary by using the definition in \Eq{def: capacity}. Furthermore, we can also provide conditions sufficient to prove the strong converse property of an upper bound on the channel capacity. In particular, Corollary \ref{cor: capacity bound} can be used together with \Eq{eq: fg Emax} in order to show that $E_{\text{max}}$ provides a strong converse bound on the capacity of a single channel, as originally proven in Ref.~\cite{Christandl_16}.

\begin{thm} \label{thm: key hypothesis}
	 If $E$ and $\mathcal N$ satisfy properties $P1$ and $P2$, the average number of ebits (pbits) generated by an adaptive protocol $\mathcal P_{\epsilon,n}$ assisted by two-way classical communication can be upper bounded as
	 \begin{equation}\label{eq26}
	 \av{\log_2 d}_{\mathcal P_{\epsilon,n}} \leq  
	 \frac{1}{g_E(\epsilon)} \left[f_E(\epsilon)  +  \av{m}_{\mathcal P_{\epsilon,n}} E\left(\mathcal N\right)\right],
	 \end{equation}
	 where $\av{m}_{\mathcal P_{\epsilon,n}}$ is the average number of times the channel has been used.
\end{thm}

\begin{cor} \label{cor: capacity bound}
	If $E$ and $\mathcal N$ satisfy properties $P1$ and $P2$, the capacity of $\mathcal N$ assisted by two-way classical communication  can be upper bounded as
	\begin{equation}\label{res: single-channel capacity bound}
	C(\mathcal N) \leq E(\mathcal N).
	\end{equation}
Furthermore, if $g_E(\epsilon)=1$ and $f_E(\epsilon)=c\log_2\frac{1}{1-\epsilon/2}$, for $c>0$, this is a strong converse bound.
\end{cor}
\begin{proof}[Proof of Corollary \ref{cor: capacity bound}]
	By definition of capacity [see \Eq{def: capacity}], only protocols using the channel a fixed number of times should be considered. Equation \eqref{res: single-channel capacity bound} is thus a straightforward consequence of  $\av{m}_{\tilde{\mathcal P}_{\epsilon,N}} = N$, $\lim_{\epsilon \to 0} g_E(\epsilon) = 1$, and $\lim_{\epsilon \to 0} f_E(\epsilon) = 0$. In order to see the strong converse property, we need to express \Eq{eq26} in terms of the error $\frac{1}{2}\Vert\rho_{AB} - \phi_{AB}(d)\Vert_1 = \epsilon/2 \in [0,1]$. Namely
\begin{equation}
\epsilon/2\geq1-2^{-\frac{1}{c}\big[\av{\log_2 d}_{\tilde{\mathcal P}_{\epsilon,N}}-N E(\mathcal{N})\big]},	
\end{equation}
which tends to $1$ exponentially fast in the number $N$ of channel uses as the rate $\frac{1}{N}\av{\log_2 d}_{\tilde{\mathcal P}_{\epsilon,N}}$ exceeds $E(\mathcal{N})$.
\end{proof}


As can be expected, a bipartite situation $A:B$ is easier to study than a scenario in which Alice and Bob need to cooperate with other nodes $\{C_i\}_{i=1}^N$ in the network in order to achieve their communication task. Building on this intuition, the authors of Refs. \cite{Azuma_16,Pirandola_Network_2016} derived upper bounds on network capacities by considering a bipartition $A\mathcal C_A:B\mathcal C_B$, and by extending the regions controlled by Alice and Bob so as to include in them also the remaining nodes on their side. Intuitively, an upper bound can be obtained in this manner because the achievable communication rate between the ``extended'' parties has to be  larger than the one achievable by the real $A$ and $B$. In this framework, any given bipartition $\{\mathcal C_A, \mathcal C_B\}$ of the network leads to a different upper bound, in which only the channels corresponding to the edges in $L_{\mathcal C_A}$ contribute. 
Although the proof that led to the result in Ref.~\cite{Azuma_16} is based on a particular choice of entanglement measure, we can see how the same reasoning applies to any entanglement measure satisfying properties $P1$ and $P2$ for any channel connecting the two network partitions. This is the result of the next theorem.

\begin{thm} \label{thm: ent-based network bound}
Consider a quantum network with an associated directed graph G. For a given bipartition $\{\mathcal C_A,\mathcal C_B\}$ of the network nodes $\{C_i\}_i$, let $L_{\mathcal C_A} \subset L$ be the set of edges in G that connect a node in $A\mathcal C_A$ with one in $\mathcal C_B B$. The average number of ebits (or pbits) that Alice and Bob share at the end of a given adaptive protocol $\mathcal P_{\epsilon,n}$, assisted by unlimited classical communication, can be upper bounded as
\begin{equation}\label{eq: generic network bound for bits}
\av{\log_2 d}_{\mathcal P_{\epsilon,n}}  \leq   
\frac{1}{g_E(\epsilon)} \left[f_E(\epsilon) + 
{\mathcal E}_E(\mathcal P_{\epsilon,n},\mathcal C_A)
\right],
\end{equation}
where
\begin{equation}\label{def: U}
{\mathcal E}_E(\mathcal P_{\epsilon,n},\mathcal C_A) \equiv \sum_{l \in L_{\mathcal C_A}} \av{m^{(l)}}_{\mathcal P_{\epsilon,n}} E\left(\mathcal N^{(l)}\right),
\end{equation}
for any entanglement measure $E$ satisfying hypotheses $P1$ and $P2$ for any channel $\mathcal N^{(l)}$ with $l\in L_{\mathcal C_A}$.
\end{thm}

At this point we can make a few comments on this bound. In virtue of Theorem~\ref{thm: key hypothesis}, we point out that the entanglement of the channel $\mathcal N^{(l)}$ has to be larger than the single-channel capacity $C(\mathcal N)$. Therefore, the gap between the two sides of \Eq{eq: generic network bound for bits} is reduced if a certain measure of entanglement can better approximate the capacity of the channels in $L_{\mathcal C_A}$.  
Furthermore, among the known entanglement measures satisfying $P1$ and $P2$ for any channel $\mathcal N^{(l)}$ with $l\in L_{\mathcal C_A}$, the best choice is to choose the one minimising ${\mathcal E}_E(\mathcal P_{\epsilon,n},\mathcal C_A)$.
If we label by $E|_{\mathcal C_A}$ the set of entanglement measures satisfying properties $P1$ and $P2$ for the channels connecting the two partitions, the following bound can be obtained: 
\begin{equation}\label{eq: network bound}
\av{\log_2 d}_{\mathcal P_{\epsilon,n}}  \leq \min_{\mathcal C_A} \min_{E|_{\mathcal C_A}} \frac{1}{g_E(\epsilon)} \left[f_E(\epsilon) + 
{\mathcal E}_E(\mathcal P_{\epsilon,n},\mathcal C_A)
\right],
\end{equation}
where we also optimised over all possible choices for $\mathcal C_A$.

\subsection{Proofs for Theorems \ref{thm: key hypothesis} and \ref{thm: ent-based network bound}}
Any single channel can be interpreted as a simple quantum network, hence we first show how Theorem~\ref{thm: key hypothesis} can be derived from Theorem~\ref{thm: ent-based network bound}, and then we prove the latter.
The ideas that will be used for these proofs are basically the same as those used in Refs.~\cite{Azuma_16,Azuma&Kato_2016,Pirandola_15,Pirandola_Network_2016,Christandl_16,Wilde_2017,TGW_nature_14,TGW_IEEE_14}.

\begin{proof}[Proof of Theorem \ref{thm: key hypothesis}]
For a single-channel scenario, the only possible bipartition $A \mathcal C_A :\mathcal C_B B$ of the network is the trivial one $A:B$. Moreover, at every round of the adaptive strategy the only channel Alice and Bob can use is $\mathcal N_{A\to B}$, which is associated with the only edge $l_0$ of the graph. 
Therefore, for all $\bk_n$  
\begin{equation}
m^{(l_0)}(\bk_n) = m(\bk_n),
\end{equation}
and the thesis of Theorem \ref{thm: ent-based network bound} simplifies to:
	 \begin{equation}
	\av{\log_2 d}_{\mathcal P_{\epsilon,n}} \leq  
	\frac{1}{g_E(\epsilon)} \left[f_E(\epsilon)  +  \av{m}_{\mathcal P_{\epsilon,n}} E\left(\mathcal N\right)\right].
	\end{equation}
\end{proof}

\begin{proof}[Proof of Theorem \ref{thm: ent-based network bound}]
In this proof we will make use of the notation introduced in Sec. \ref{sec: adaptive strategy} to describe a generic adaptive protocol $\mathcal P_{\epsilon,n}$. Property $P1$, together with \Eq{eq: epsilon close}, implies:
\begin{equation}
\log_2 d_{\bk_n} \leq \frac{1}{g_E(\epsilon)}\left(f_E(\epsilon) + E^{A:B}(\rho_{AB}^{\bk_n})\right).
\end{equation}
By exploiting the monotonicity of $E$ under partial trace, and by averaging over all possible outcomes,
we can write for any bipartition $\{\mathcal C_A,\mathcal C_B\}$ of the set of nodes $\{C_i\}_i$:
\begin{equation} \label{eq: continuity inequality generic E}
\av{\log_2 d}_{\mathcal P_{\epsilon,n}} = \sum_{\bk_n} p(\bk_n) \log_2 d_{\bk_n} \leq 
\frac{1}{g_E(\epsilon)}\left[f_E(\epsilon) + \sum_{\bk_n} p(\bk_n) E^{A\mathcal C_A:\mathcal C_B B}(\rho^{\bk_n})  \right],
\end{equation}
where $\rho^{\bk_n}$ is the final state of the protocol, spread across the whole network. The second term written between square brackets on the right-hand side can be expanded into two terms as
\begin{equation}\label{eq: last step}
\sum_{\bk_n} p(\bk_n) E^{A\mathcal C_A:\mathcal C_B B}\left(\rho^{\bk_n}\right) \leq 
\sum_{\bk_{n-1}} p(\bk_{n-1}) E^{A\mathcal C_A:\mathcal C_B B}\left(\rho^{\bk_{n-1}}\right) + 
\sum_{\bk_{n}} p(\bk_{n}) \sum_{l \in L_{\mathcal C_A}} \delta_{l,l_{\bk_{n-1}}} E[\mathcal N^{(l)}].
\end{equation}
The former is self-similar, but evaluated on the previous round of the protocol, while the latter characterises the ability of the last channel used to create entanglement across the bipartition $A\mathcal C_A:\mathcal C_B B$. In particular, the second term does not always appear, because the channel $\mathcal N^{(l_{\bk_{n-1}})}$ might not connect $A\mathcal C_A$ with $\mathcal C_B B$, or the parties may have decided not to use a channel at all. This last case could be represented, for example, by any value of  $l_{\bk_{n-1}}$ not in the set $L$ of graph edges.
In order to prove \Eq{eq: last step}, we can first expand the left-hand side as
\begin{align} \label{eq: step proof}
\sum_{\bk_n} p(\bk_n) E^{A\mathcal C_A:\mathcal C_B B}\left(\rho^{\bk_n}\right) 
& = \sum_{\bk_{n-1}} p(\bk_{n-1}) \left[\sum_{k_n} p(k_n|\bk_{n-1}) E^{A\mathcal C_A:\mathcal C_B B}\left(\rho^{\bk_n}\right)\right], 
\end{align}
and then use the following chain of inequalities:
\begin{equation} \label{eq: ent bound}
\sum_{k_n} p(k_n|\bk_{n-1}) E^{A\mathcal C_A:\mathcal C_B B}\left(\rho^{\bk_n}\right) \stackrel{\scriptstyle (\rm{i})}{\leq} E^{A\mathcal C_A:\mathcal C_B B}\left(\tilde\rho^{\bk_{n-1}}\right) \stackrel{\scriptstyle (\rm{ii})}{\leq} E^{A\mathcal C_A:\mathcal C_B B}\left(\rho^{\bk_{n-1}}\right) + \sum_{l \in L_{\mathcal C_A}} \delta_{l,l_{\bk_{n-1}}} E[\mathcal N^{(l)}],
\end{equation}
where $(\rm{i})$ is due to the monotonicity of $E$ under LOCC operations, while $(\rm{ii})$ directly follows from property $P2$.
After combining Eqs.~\eqref{eq: step proof} and \eqref{eq: ent bound}, we can recover \Eq{eq: last step} simply by noticing that the average over $\bk_{n-1}$ on the rightmost term of \Eq{eq: ent bound} can be freely changed into an average over $\bk_n$.
The same procedure can be iteratively applied for every round of the protocol, so that at the end we are left with
\begin{align}
\sum_{\bk_n} p(\bk_n) E^{A\mathcal C_A:\mathcal C_B B}\left(\rho^{\bk_n}\right) & \leq 
E^{A\mathcal C_A:\mathcal C_B B}\left(\rho^{(1)}\right) + 
\sum_{j=0}^{n-1} \sum_{\bk_{n}} p(\bk_{n}) \sum_{l \in L_{\mathcal C_A}} \delta_{l,l_{\bk_{j}}} E[\mathcal N^{(l)}] \notag \\
& = \sum_{l \in L_{\mathcal C_A}} \av{m^{(l)}} E[\mathcal N^{(l)}], 
\end{align}
where the last equality is due to the separability of the initial state $\rho^{(1)}$ and to the definition of $\av{m^{(l)}}$ given in \Eq{def: m for a channel}. At this point, the thesis of Theorem~\ref{thm: ent-based network bound} follows directly from the inequality given in \Eq{eq: continuity inequality generic E}.
\end{proof}

\section{Versatile upper bound for quantum networks}\label{sec: hybrid bound}

As we have seen, an entanglement measure $E$ can lead to an upper bound on the capacity of a channel if it satisfies a continuity inequality (property $P1$), and a recursive relation connecting the entanglement of the state before and after the channel application (property $P2$). In the previous section we discussed the possibility of changing entanglement measures across different bipartitions. However, in doing so we have to guarantee that, for each bipartition $\mathcal C_A$, the chosen measure satisfies property $P2$ for \emph{every} channel $\mathcal N^{(l)}$ with $l\in L_{\mathcal C_A}$. 
This constraint leads to weaker upper bounds than what would be obtained if we could change entanglement measure on a channel-by-channel basis. 
For example, consider a situation where all the channels in a given bipartition are Choi-simulable, with only one exception: the presence of this single unsimulable channel prevents us from using $E_{\text{R}}$ in the bound of Theorem \ref{thm: ent-based network bound}. Instead, we are forced to use some broadly applicable entanglement measure (as $E_{\text{sq}}$ or $E_{\text{max}}$) on every channel of the bipartition, thus loosening the bound.

In this section we overcome this issue, by exploiting a recent result on sandwiched R\'{e}nyi entropies \cite{Christandl_16}. In particular, we construct an upper bound on $\av{\log_2 d}_{\mathcal P_{\epsilon,n}}$ that allows us to switch between $E_{\text{R}}$ and $E_{\text{max}}$, depending on the Choi-simulability of each channel. 
To begin with, in the following we describe the recent result obtained in Ref. \cite{Christandl_16}, which is the cornerstone of our method. Then, we prove our main result.

\subsection{Versatile Property P2 for the relative and max-relative entropy of entanglement}
For any quantum channel $\mathcal N_{A^\prime \to B^\prime}$, and any real parameter $1\leq \alpha <\infty$, if
\begin{equation}
\tilde \rho_{AB^\prime B} = \mathcal N_{A^\prime \to B^\prime}(\rho_{AA^\prime B}),
\end{equation}
one has \cite{Christandl_16}
\begin{equation}\label{eq: generic Emax bound}
E_\alpha(\tilde\rho_{A B^\prime B}) \leq E_{\text{max}}(\mathcal N_{A^\prime \to B^\prime}) + E_\alpha(\rho_{A A^\prime B}).
\end{equation}
The quantity $E_\alpha$ is defined in terms of the sandwiched R\'{e}nyi relative entropy $\tilde D_\alpha$ \cite{Tomamichel_13,Wilde2014_entropy}:
\begin{align} \label{def: E alpha}
E_\alpha(\rho_{AB}) &= \min_{\sigma_{AB} \in \text{SEP}} \tilde D_{\alpha}(\rho_{AB}\Vert \sigma_{AB}) 
 \notag\\ 
&= \min_{\sigma_{AB} \in \text{SEP}} 
\frac{1}{\alpha-1}\log_2 \Tr{\left(\sigma_{AB}^{\frac{1- \alpha}{2\alpha}} \;\rho_{AB}\; \sigma_{AB}^{\frac{1- \alpha}{2\alpha}}\right)^\alpha},
\end{align}
where $\sigma_{AB}$ is optimised over all separable states.
As $E_\alpha$ tends respectively to $E_{\text{R}}$ and $E_{\text{max}}$ in the limits of $\alpha \to 1$ and $\alpha \to \infty$, by setting $\alpha = 1$ in \Eq{eq: Emax bound} we obtain
\begin{equation}\label{eq: Emax bound}
E_{\text{R}}(\tilde\rho_{A B^\prime B}) \leq E_{\text{max}}(\mathcal N_{A^\prime \to B^\prime}) + E_{\text{R}}(\rho_{A A^\prime B}).
\end{equation}
This inequality closely resembles property $P2$ for $E_{\text{R}}$, which was obtained in \Eq{eq: ER bound} for Choi-simulable channels. However, thanks to the introduction of $E_{\text{max}}$ on the right hand side, \Eq{eq: Emax bound} now holds even for non Choi-simulable channels.
By combining \Eq{eq: ER bound} with \Eq{eq: Emax bound}, we can obtain a versatile property $P2$ for the relative entropy of entanglement, in which the right-hand side changes according to the Choi-simulability of $\mathcal N_{A^\prime \to B^\prime}$:
\begin{equation} \label{eq: single-channel bound}
E_{\text{R}}(\tilde\rho_{A B^\prime B}) \leq E_{\text{R}}(\rho_{A A^\prime B}) + 
\begin{cases}
E_{\text{R}}\left(\mathcal N\right), & \text{if } \mathcal N \in \mathcal S,\\
 E_{\text{max}}\left(\mathcal N\right), & \text{otherwise},
\end{cases}
\end{equation}
where $\mathcal S$ is the set of Choi-simulable channels. Note that this is the best choice, as $E_{R} \leq E_{\text{max}}$ for all states [see \Eq{eq: Er Emax relation}].

\subsection{Versatile upper bound for quantum networks}\label{sec: new versatile bound}
We have now all the tools to obtain a versatile upper bound on the length of ebit (or pbits) shared by Alice and Bob at the end of a generic adaptive protocol $\mathcal P_{\epsilon,n}$, assisted by unlimited classical communication, over a quantum network.
\begin{thm} \label{thm: network bound}
	Consider a quantum network with an associated directed graph G. For a given bipartition $\{\mathcal C_A,\mathcal C_B\}$ of the network nodes $\{C_i\}_i$, let $L_{\mathcal C_A} \subset L$ be the set of edges in G that connect a node in $A\mathcal C_A$ with one in $\mathcal C_B B$. The average number of ebits (or pbits) that Alice and Bob share at the end of a given adaptive communication protocol $\mathcal P_{\epsilon,n}$, assisted by unlimited classical communication, can be upper bounded as
	\begin{equation}\label{eq: heterogeneous network bound}
	\av{\log_2 d}_{\mathcal P_{\epsilon,n}}  \leq 
	\frac{1}{g_{E_{\rm{R}}}(\epsilon)} \left[f_{E_{\rm{R}}}(\epsilon) + 
	{\mathcal E}^\prime(\mathcal P_{\epsilon,n}, \mathcal C_A)
	\right],
	\end{equation}
	where
	\begin{equation}\label{def: Uprime}
	{\mathcal E}^\prime(\mathcal P_{\epsilon,n}, \mathcal C_A) \equiv \sum_{\substack{l \in L_{\mathcal C_A}:\\\mathcal N^{(l)}\in\mathcal S}} \av{m^{(l)}}_{\mathcal P_{\epsilon,n}} 
	E_{\rm{R}}\left(\mathcal N^{(l)}\right) + 
	\sum_{\substack{l \in L_{\mathcal C_A}:\\\mathcal N^{(l)}\notin\mathcal S}} \av{m^{(l)}}_{\mathcal P_{\epsilon,n}} 
	E_{\rm{max}}\left(\mathcal N^{(l)}\right),
	\end{equation}
	with $f_{E_{\text{R}}}(\epsilon) = -2 [\epsilon\log_2 \epsilon + (1-\epsilon)\log_2 (1-\epsilon)]$ and $ g_{E_{\rm{R}}}(\epsilon) = 1-8\epsilon$.

\end{thm}
\begin{proof}[Proof of Theorem \ref{thm: network bound}]
	The proof follows closely the one provided for Theorem \ref{thm: ent-based network bound}, with $E = E_{\text{R}}$. The only difference lies in \Eq{eq: ent bound}, where we use the inequality in  \Eq{eq: single-channel bound} instead of the original property $P2$. Therefore, \Eq{eq: ent bound} has to be substituted with
	\begin{equation} 
	E_{\text{R}}^{A\mathcal C_A:\mathcal C_B B}\left(\tilde\rho^{\bk_{n-1}}\right) \leq E_{\text{R}}^{A\mathcal C_A:\mathcal C_B B}\left(\rho^{\bk_{n-1}}\right) + \sum_{\substack{l \in L_{\mathcal C_A}:\\\mathcal N^{(l)}\in\mathcal S}} \delta_{l,l_{\bk_{n-1}}} E_{\text{R}}\left(\mathcal N^{(l)}\right)
	+ \sum_{\substack{l \in L_{\mathcal C_A}:\\\mathcal N^{(l)}\notin\mathcal S}} \delta_{l,l_{\bk_{n-1}}} E_{\text{max}}\left(\mathcal N^{(l)}\right),
	\end{equation}
	where we split the sum over the Choi-simulable and non-Choi-simulable channels connecting the nodes on different sides of the network partition.
	The remainder of the proof then follows the same steps used in the proof of Theorem~\ref{thm: ent-based network bound}. We also explicitly provide the expressions for the functions $f_{E_{\text{R}}}(\epsilon)$ and $g_{E_{\text{R}}}(\epsilon)$ appearing in Property $1$ (see e.g. Ref.~\cite{Pirandola_15}).
	\end{proof}

Thanks to this result, we have managed to merge the upper bounds based on the quantities ${\mathcal E}_{E_{\text{R}}}$ and ${\mathcal E}_{E_{\text{max}}}$ into a single bound, which retains the advantages given by the two entanglement measures, i.e., tightness and broad applicability. Therefore, in assessing the communication performance of an adaptive protocol $\mathcal P_{\epsilon,n}$ over a quantum network, for any given bipartition $A\mathcal C_A:\mathcal C_B B$ one just needs to compare ${\mathcal E}^\prime$ with the bound ${\mathcal E}_{E_{\text{sq}}}$ based on the squashed entanglement \cite{Azuma_16}. This is because the dependence on $f_E$ and $g_E$ vanishes for small errors $\epsilon$.
In particular, the advantage of using ${\mathcal E}^\prime$ over ${\mathcal E}_{E_{\text{sq}}}$ for the bipartition $A\mathcal C_A:\mathcal C_B B$ can be quantified by the parameter
\begin{equation}\label{def: mu parameter}
\mu_{\mathcal C_A}(\mathcal P_{\epsilon,n}) = \frac{{\mathcal E}_{E_{\text{sq}}}(\mathcal P_{\epsilon,n}, \mathcal C_A) - {\mathcal E}^\prime(\mathcal P_{\epsilon,n}, \mathcal C_A) }{{\mathcal E}_{E_{\text{sq}}}(\mathcal P_{\epsilon,n}, \mathcal C_A) + {\mathcal E}^\prime(\mathcal P_{\epsilon,n}, \mathcal C_A)},
\end{equation}
which is defined in the range $[-1,+1]$ and is positive when the versatile bound ${\mathcal E}^\prime$ is tighter than ${\mathcal E}_{E_{\text{sq}}}$.
The sign of $\mu_{\mathcal C_A}(\mathcal P_{\epsilon,n})$ will ultimately depend on the details of the bipartition and on the average number of times each channel is used. However, we can expect ${\mathcal E}^\prime$ to be tighter than ${\mathcal E}_{E_{\text{sq}}}$ on bipartitions mostly connected by Choi-simulable channels, because the most common of these channels satisfy $E_{\text{R}}(\mathcal N) < E_{\text{sq}}(\mathcal N)$. In contrast, when there is a considerable amount of channels that are not Choi-simulable, the sign of $\mu_{\mathcal C_A}(\mathcal P_{\epsilon,n})$ will strongly depend on the sign of $E_{\text{sq}}(\mathcal N) - E_{\text{max}}(\mathcal N)$: every non Choi-simulable channel $\mathcal N$ for which this difference is positive will enhance the usefulness of ${\mathcal E}^\prime$ over ${\mathcal E}_{E_\text{sq}}$.

We should stress that ${\mathcal E}^\prime$,  ${\mathcal E}_{E_\text{sq}}$, and thus $\mu_{\mathcal C_A}(\mathcal P_{\epsilon,n})$ might not be easily evaluated, because the exact values of $E_{\text{sq}}(\mathcal N)$ and $E_{\text{max}}(\mathcal N)$ are not known for many channels. When evaluating communication bounds, in practice it is common to consider the smallest known upper bounds $\tilde E_{\text{sq}}(\mathcal N)$ and $\tilde E_{\text{max}}(\mathcal N)$ on those unknown quantities, rather than their exact values. 
When these approximations are introduced in Eqs.~\eqref{def: U} and \eqref{def: Uprime} we are left with slightly different quantities  $\tilde {\mathcal E}^\prime$ and $\tilde{\mathcal E}_{E_\text{sq}}$, which if used instead of ${\mathcal E}^\prime$ and ${\mathcal E}_{E_\text{sq}}$ in \Eq{def: mu parameter} lead to a modified parameter $\tilde\mu_{\mathcal C_A}(\mathcal P_{\epsilon,n})$.
Then, we can say that currently our versatile upper bound yields a better result than the network bound based on squashed entanglement when $\tilde\mu_{\mathcal C_A}(\mathcal P_{\epsilon,n}) > 0$.

Before discussing examples of networks where the bound provided by Theorem~\ref{thm: network bound} becomes tighter than its counterpart based on the squashed entanglement, we first need to evaluate $E_{\text{max}}(\mathcal N)$ for some channels of interest.
In particular, in the next section we will consider typical qubit quantum channels.

\section{Max-relative entropy of entanglement of qubit channels}\label{sec: channels}

In this section we develop a method to obtain lower and upper bounds on the max-relative entropy of entanglement of channels invariant under phase rotations, and to evaluate $E_{\rm{max}}$ itself for Choi-simulable channels with the same symmetry. After that, we discuss the possibility of using semidefinite programming (SDP) in order to evaluate the max-relative entropy of entanglement of qubit channels, by using a formulation recently introduced in Ref.~\cite{berta2017amortization}.
Interestingly, by combining these tools we are able to analytically obtain the max-relative entropy of entanglement of the qubit amplitude damping channel $\mathcal N^{(\text{ad})}$. As this channel is not Choi-simulable its capacity is still unknown, although several upper bounds on it have been recently derived \cite{Pirandola_15,Goodenough_16}. At the end of this section we also numerically evaluate the max-relative entropy of entanglement of other common Choi-simulable qubit channels: dephasing, erasure and depolarising channels.
Although the relative entropy of entanglement could be used to bound the capacities of these channels, the purpose of this analysis is to see how far off the upper bound based on max-relative entropy of entanglement is, compared with other bounds known in the literature.

In general, the calculation of the max-relative entropy of entanglement of a channel involves a max-min optimisation [see Eqs. \eqref{def: Dmax} and \eqref{def: Emax}]:
\begin{align}\label{eq: Emax explicit}
E_{\text{max}}(\mathcal N) 
&= \max_{\rho_{A A^\prime}} \min_{\sigma_{AB}\in \text{SEP}} \inf_x\{x \in\mathbb{R} \vert 2^x \sigma_{AB} - \mathcal N_{A^\prime\to B}[\rho_{AA^\prime}]\geq 0 \}.
\end{align}
In fact, the maximisation over $\rho_{AA^\prime}$ can be restricted to bipartite pure states with the dimension of $A$ equal to that of $A^\prime$. This can be shown by purifying $\rho_{AA^\prime}$ and by applying the Schmidt decomposition and the date processing inequality for the sandwiched R\'{e}nyi relative entropy \cite{Beigi2013}. Nonetheless, typically the optimisation leading to $E_{\text{max}}(\mathcal N)$ is still not trivial to perform.
However, the max-relative entropy of entanglement of a channel can always be bounded from both sides as stated in the following proposition, whose proof can be found in Appendix~\ref{app: proposition proofs}. The upper bound is a re-elaborated version of the upper bound on the max-relative entropy of entanglement of a channel studied in Ref. \cite{Christandl_16}.
In order to explicitly perform the required optimisations, it is useful to exploit as much as possible the symmetries of the considered channel $\mathcal N$.  In particular, in Appendix~\ref{app: tech tools} we develop tools applicable to qubit channels invariant under phase rotations.

\begin{prop}\label{prop: bounds on Emax}
	Let $\pi_{\mathcal N} = \Id_{A}\otimes\mathcal N_{A^\prime \to B}[\psi_{AA^\prime}]$ be the Choi-Jamio\l{}kowski state associated with the quantum channel $\mathcal N$ with input dimension $d$, where $\psi_{AA^\prime}$ is a maximally entangled state. Then, we have
	\begin{equation} \label{eq: bounds on Emax}
	\min_{\sigma_{AB}\in \rm{SEP}}  D_{\rm{max}}(\pi_{\mathcal N}\Vert\sigma_{AB}) \leq 
	E_{\rm{max}}(\mathcal N)  
	\leq \min_{\substack{\sigma_{AB}\in \rm{SEP} \\ {\rm{Tr}}_{B}[\sigma_{AB}]=\Id_A/d}}  D_{\rm{max}}(\pi_{\mathcal N}\Vert\sigma_{AB}). 
	\end{equation}	
	Moreover, if $\mathcal N$ is Choi-simulable, the lower bound is equal to $E_{\rm{max}}(\mathcal N)$ itself.
\end{prop} 

An alternative expression for the max-relative entropy of a channel has  been recently proposed in Ref.~\cite{berta2017amortization}, and can be written as
\begin{equation}\label{eq: SDP Emax}
E_{\text{max}}(\mathcal N)=\log_2\Sigma(\mathcal N),
\end{equation}
where
\begin{equation}\label{eq:Sigma}
\Sigma(\mathcal N)=\min_{Y_{AB}\in \overrightarrow{\rm{SEP}}}\left\{\left\Vert{\rm{Tr}}_{B}[Y_{AB}]\right\Vert_\infty:Y_{AB} - d \, \pi_{\mathcal N}\geq0\right\}.
\end{equation}
Here $d$ is the input dimension of the channel $\mathcal N$, and $\overrightarrow{\rm{SEP}}$ denotes the cone of (unnormalised) separable operators, i.e., the set of all operators $X_{AB}$ that can be decomposed as $X_{AB}=\sum_{i=1}^L P_A^i\otimes Q_B^i$ for some positive integer $L$ and positive semidefinite operators $P_A^i$ and $Q_B^i$. Note that for qubit channels we can replace $\overrightarrow{\rm{SEP}}$ by the cone of all positive semidefinite operators that are PPT, thus making the evaluation of \Eq{eq:Sigma} efficiently computable via SDP.

\subsection{Amplitude damping channel}
We begin by studying the most important example among channels that are not Choi-simulable: the qubit amplitude damping channel $\mathcal N^{(\text{ad})}_\lambda$, which can be written as 
\begin{equation}
\mathcal N^{(\text{ad})}_\lambda (\rho) = \sum_{i=1}^2 M_i(\mathcal N^{(\text{ad})}_\lambda) \rho M_i^\dagger(\mathcal N^{(\text{ad})}_\lambda),
\end{equation}
in terms of the Kraus operators:
\begin{equation}
M_1(\mathcal N^{(\text{ad})}_\lambda) = \ketbra{0}{0} + \sqrt{1-\lambda} \ketbra{1}{1}, \qquad M_2 (\mathcal N^{(\text{ad})}_\lambda) = \sqrt{\lambda}\ketbra{0}{1}.
\end{equation}
Note that $\mathcal N^{(\text{ad})}_\lambda$ reduces to the identity channel when $\lambda = 0$.
In particular, we analytically calculate the lower and upper bounds on $E_{\text{max}}(\mathcal N^{(\text{ad})}_\lambda)$ found in Proposition \ref{prop: bounds on Emax}:
\begin{equation} \label{eq: explicit bounds on amplitude damping}
F(\lambda)
\leq E_{\rm{max}}(\mathcal N^{(\rm{ad})}_\lambda)\leq \tilde E_{\rm{max}}(\mathcal N^{(\rm{ad})}_\lambda),
\end{equation}
where
\begin{equation}
F(\lambda) \equiv \begin{cases}
\log_2\left[\frac{1}{2}(1+\sqrt{1-\lambda})^2\right], & \rm{if }\lambda\leq \frac{\sqrt{5}-1}{2},\\
\log_2\left(\frac{1+\lambda}{2\lambda}\right), &\rm{if } \lambda \geq \frac{\sqrt{5}-1}{2},
\end{cases}
\qquad \text{and} \qquad
\tilde E_{\rm{max}}(\mathcal N^{(\rm{ad})}_\lambda) \equiv \log_2\left(2-\lambda\right).
\end{equation}
The proofs for these inequalities can be found respectively in Appendices \ref{app: upper bound amp damp} and \ref{app: lower bound proof}.
We stress that $\tilde E_{\rm{max}}(\mathcal N^{(\rm{ad})}_\lambda)$ is also an upper bound on the capacity $C\big(\mathcal N^{(\text{ad})}_\lambda\big)$, whereas $F(\lambda)$ does not have any known relation with the capacity. 
Interestingly, the numerical evaluation of $E_{\text{max}}(\mathcal N^{(\lambda)})$ via the SDP procedure in \Eq{eq:Sigma} coincides with the upper bound found in \Eq{eq: explicit bounds on amplitude damping} up to numerical errors. This suggests that for all $\lambda \in [0,1]$ the max-relative entropy of entanglement of the amplitude damping channel exactly coincides with its upper bound found through Proposition \ref{prop: bounds on Emax}.
Indeed, this is analytically proven in Appendix~\ref{app: Emax Amd Damp}, and we can write it here as a proposition.
\begin{prop}\label{prop: B UB for ampd damp}
	The max-relative entropy of entanglement of a qubit amplitude damping channel $\mathcal N^{(\rm{ad})}_\lambda$ is 
	\begin{equation}
	E_{\rm{max}}\big(\mathcal N^{(\rm{ad})}_\lambda\big) = \log_2(2-\lambda).
	\end{equation}
\end{prop}

The plot in Fig.~\ref{fig: amp damp bounds} shows how $E_{\text{max}}(\mathcal N^{(\rm{ad})}_\lambda)$, plotted as a black curve, can be compared with other bounds on $C(\mathcal N^{(\rm{ad})}_\lambda)$ known in the literature. In particular, it is much smaller than the upper bound on the capacity obtained in Ref. \cite{Pirandola_15}, represented by the top blue solid curve in Fig.~\ref{fig: amp damp bounds}.
The latter was obtained by decomposing the amplitude damping channel as $\mathcal N^{(\text{ad})}_\lambda = \mathcal N_1 \circ \mathcal L \circ  \mathcal N_2$, where $\mathcal L \in \mathcal S$ but $\mathcal N_1$ and $\mathcal N_2$ are not, and by considering the bound $C(\mathcal N^{(\text{ad})}_\lambda) \leq E_{\text{R}}(\mathcal L)$. 
However, the upper bound on the capacity based on the squashed entanglement \cite{Pirandola_15} is smaller than our result obtained through $E_{\text{max}}$.
For completeness, we also plotted the best known lower bound on $C(\mathcal N^{(\text{ad})}_\lambda)$, which narrows the region where the capacity value could be \cite{Pirandola_15,Goodenough_16}.
From this analysis, we can conclude that at the moment the best known upper bound on the capacity of the amplitude damping channel remains based on its squashed entanglement. 

\begin{figure}
	\centering
	\includegraphics[scale = 0.7]{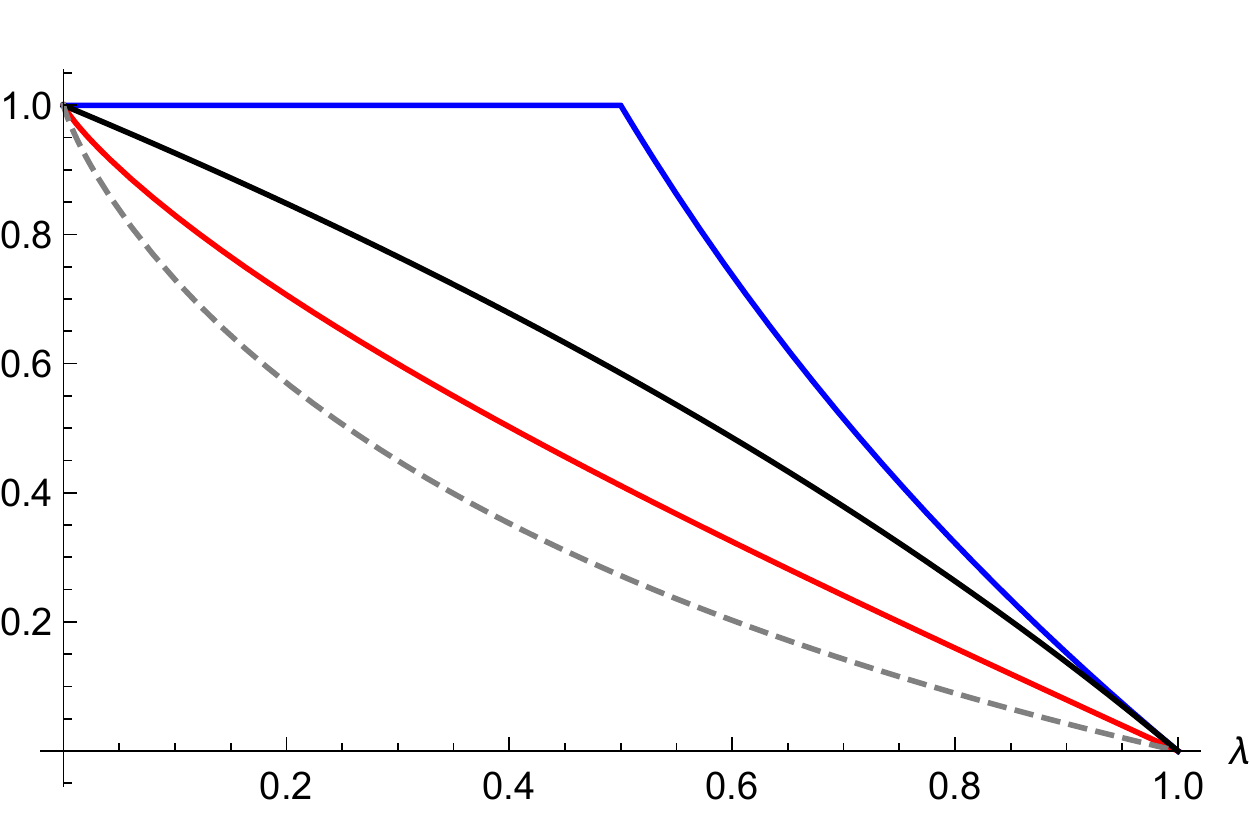}
	\caption{Grey dashed line: Best known lower bound on the capacity $C(\mathcal N^{(\text{ad})})$, corresponding to the reverse coherent information \cite{Devetak_RevCohInfo_2006} of the channel (see \cite{Pirandola_15}). The remaining solid lines are all upper bounds on the capacity $C(\mathcal N^{(\text{ad})})$. In particular, the top blue line is the bound based on the relative entropy of entanglement discussed in Ref.~\cite{Pirandola_15}, the bottom red line is the best known bound based on the squashed entanglement \cite{Pirandola_15}, and the black line in the middle is $E_{\text{max}}(\mathcal N^{(\text{ad})})$. \label{fig: amp damp bounds}}
\end{figure}

\subsection{Other Choi-simulable channels}\label{sec: other channels}
Here we numerically evaluate the max-relative entropy of entanglement of some common qubit channels: dephasing, erasure, and depolarising channels. Note that the capacities of the first two channels are given by single-letter formulas, and are thus known exactly. In our numerical simulations we perform the SDP optimisation in \Eq{eq:Sigma}, which yields the same results obtained by numerically evaluating the lower bound in Proposition \ref{prop: bounds on Emax}.

The dephasing channel $\mathcal N^{\text{(deph)}}_{\lambda}$ and depolarising channel $\mathcal{N}^{\text{(depo)}}_\lambda$ can be respectively written in terms of a set of $2$ and $5$ Kraus operators:
\begin{equation}
M_1(\mathcal N^{\text{(deph)}}_\lambda) = \sqrt{1-\frac{\lambda}{2}}\Id, \qquad M_2 (\mathcal N^{\text{(deph)}}_\lambda) = \sqrt{\frac{\lambda}{2}} \sigma_z,
\end{equation}
\begin{equation}
M_{0}(\mathcal N^{\text{(depo)}}_\lambda) = \sqrt{1-\lambda}\Id, \qquad M_{ij}(\mathcal N^{\text{(depo)}}_\lambda) = \sqrt{\frac{\lambda}{2}}\ketbra{i}{j},
\end{equation}
with $i,j = 0,1$. The erasure channel $\mathcal{N}^{\text{(er)}}$, on the other hand,  is characterised by the Kraus operators
\begin{equation}
M_2(\mathcal{N}^{\text{(er)}}_\lambda) = \sqrt{1-\lambda}\Id, \qquad M_i (\mathcal{N}^{\text{(er)}}_\lambda) = \sqrt{\lambda} \ketbra{e}{i},
\end{equation}
where $i=0,1$, and $\ket{e}$ is an error state orthogonal to both $\ket{0}$ and $\ket{1}$. All these channels reduce to the identity channel when $\lambda = 0$.

We point out that exact values for the max-relative entropy of entanglement of these channels are not needed when evaluating the versatile network bound of Theorem~\ref{thm: network bound}. This is because they are all Choi-simulable, and the entanglement generated by them can be quantified by means of $E_{\text{R}}$. Nonetheless, we numerically evaluated $E_{\text{max}}(\mathcal N)$ for these channels in order to see whether the obtained values could be smaller than their counterparts based on the squashed entanglement.  
The results can be seen in Fig.~\ref{fig: Choi-simulable plots}. In all these cases $E_{\text{R}}$ yields the tighter upper bound on the capacity, followed by
the squashed entanglement, while $E_{\text{max}}$ provides the loosest bound.

\begin{figure}
	\centering
	\subfloat[Dephasing channel. \label{fig: dephasing}]{\includegraphics[scale = 0.4]{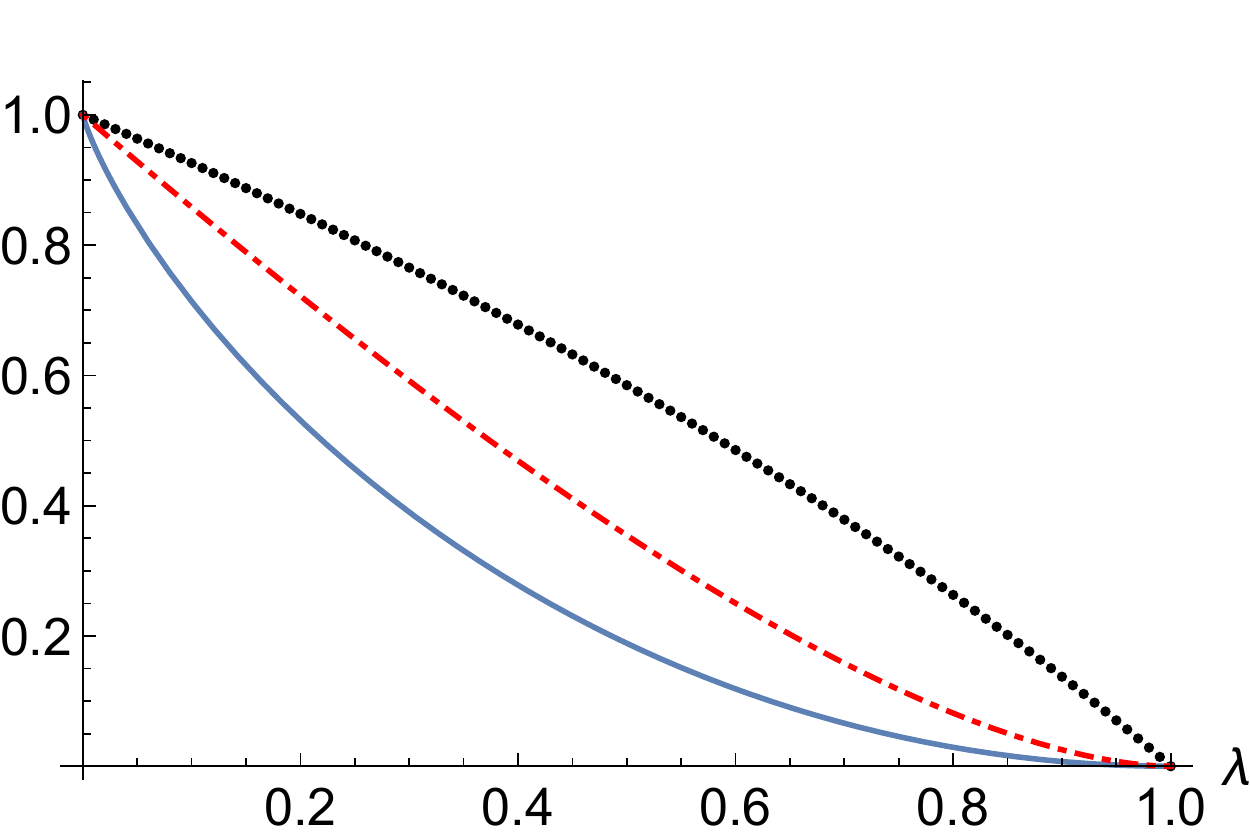}} $\quad$
	\subfloat[Erasure channel.]{\includegraphics[scale = 0.4]{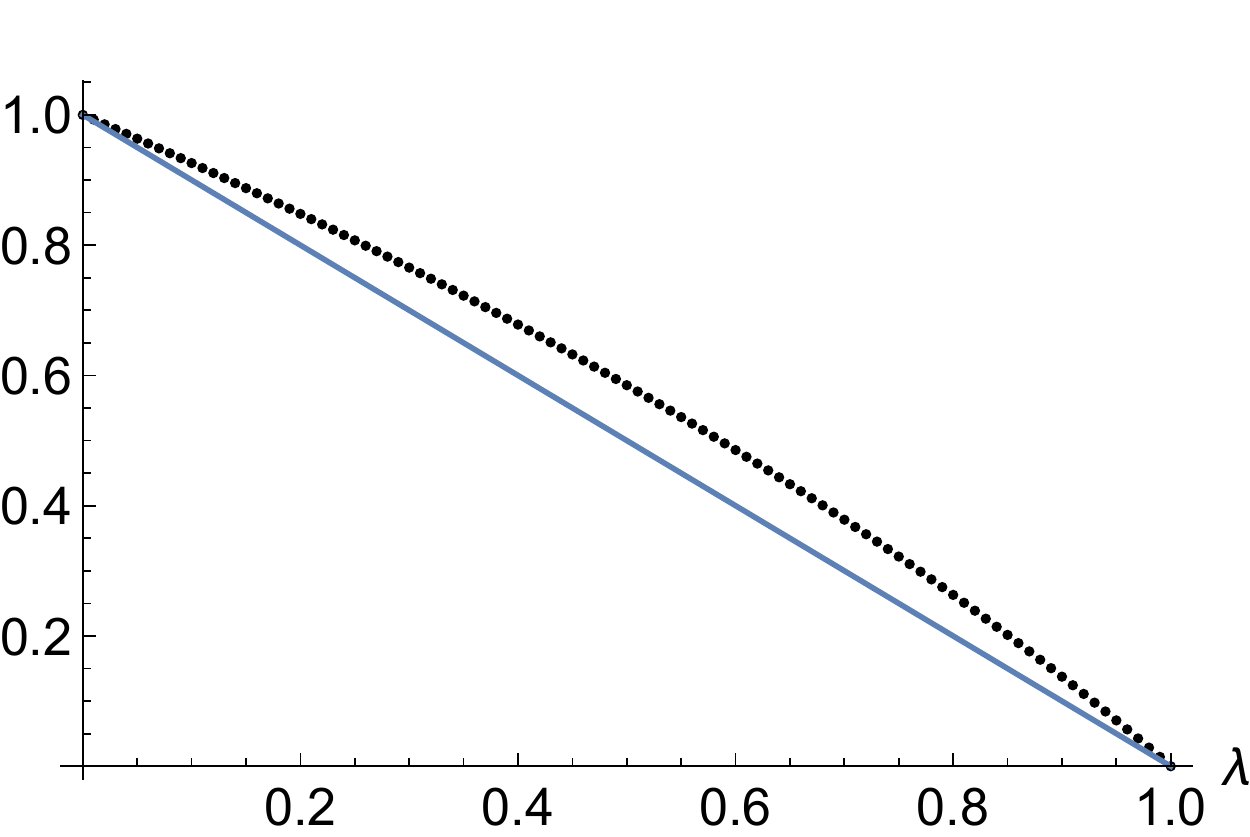}} $\quad$
	\subfloat[Depolarising channel.]{\includegraphics[scale = 0.4]{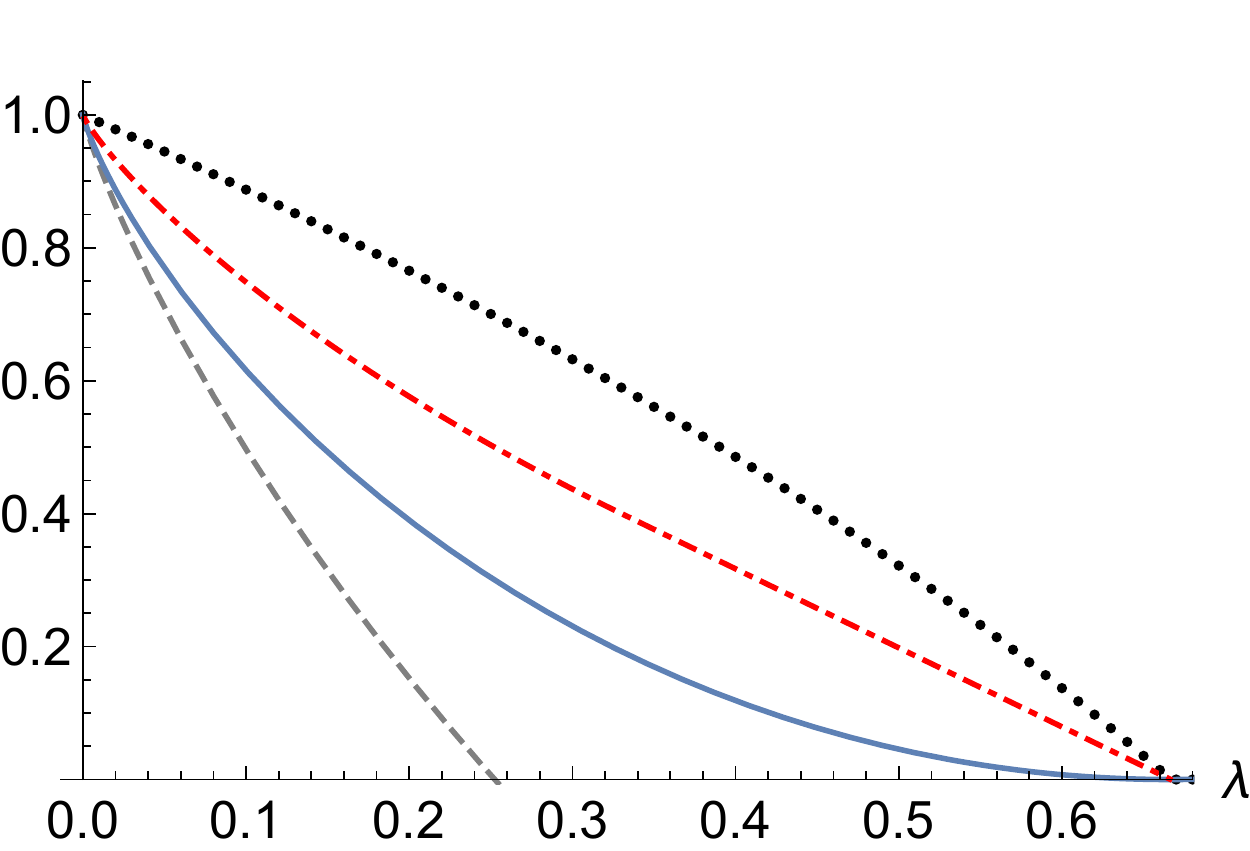}}
	\caption{Bounds on the capacity of three Choi-simulable qubit channels. In each plot, the solid blue line represents $E_{\text{R}}(\mathcal{E}_\lambda)$, and coincides with the capacity $C(\mathcal{E}_\lambda)$ for the dephasing and erasure channel. For the depolarising channel, the capacity $C\big(\mathcal{E}^{\text{(depo)}}_{\lambda}\big)$ lies between the blue solid line and the gray dashed line, which respectively represent its best known upper and lower bounds (see, e.g., Ref. \cite{Pirandola_15}). The depolarising channel has zero capacity for $\lambda>2/3$, where it becomes entanglement breaking, so that region has not been plotted. Red dot-dashed lines: smallest known upper bound on the squashed entanglement of the channels \cite{TGW_IEEE_14, Goodenough_16}. In the specific case of the erasure channel, one has $E_{\text{sq}}(\mathcal{E}^{(\text{er})}_\lambda) = E_{\text{R}}(\mathcal{E}^{(\text{er})}_\lambda)$ \cite{Goodenough_16,Pirandola_15}. Black dots: numerical evaluations of $E_{\text{max}}(\mathcal N_\lambda)$, obtained via the SDP optimisation in \Eq{eq:Sigma}.
		\label{fig: Choi-simulable plots}
	}
\end{figure}

\section{Examples}\label{sec: examples}
As we already mentioned in Sec.~\ref{sec: new versatile bound}, in order to assess whether Theorem~$\ref{thm: network bound}$ leads to a tighter bound than the version of Theorem~\ref{thm: ent-based network bound} based on the squashed entanglement, for any considered bipartition of the network one should study the sign of the parameter $\tilde \mu_{\mathcal C_A}$. 
This can be found as in \Eq{def: mu parameter}, but substituting $E_{\text{sq}}(\mathcal N)$ with its best known upper bound available in the literature.
In what follows we provide two examples where $\tilde \mu_{\mathcal C_A} >0$.

At first, we should stress that there are quantum channels with  $E_{\text{sq}}(\mathcal N)$ much larger than $E_{\text{max}}(\mathcal N)$. An example are the ``flower channels'' \cite{Christandl_05_flower,Horodecki_05_flower} for which the gap between these two quantities can increase with the dimension of the input system \cite{Christandl_16}. This is due to the fact that the squashed entanglement is ``lockable'', which means that by tracing out a subsystem of dimension $d$ its value can change by an amount more than logarithmic in $d$. On the contrary, $E_{\text{max}}$ is not lockable, and it does not suffer from this drawback. Therefore, ${\mathcal E}^\prime$ would be much tighter than ${\mathcal E}_{E_{\text{sq}}}$ when evaluated on bipartitions mostly composed by flower channels, or composed by flower channels and Choi-simulable channels with $E_{\text{R}}$ smaller than $E_{\text{sq}}$, as the qubit channels studied in Sec.~\ref{sec: other channels}.
However, it could be argued that this example is rather artificial, and it is not likely to appear in any realistic communication scenario. For this reason, we also consider a more practical example where the two components of a bipartition $A\mathcal C_A:\mathcal C_BB$ are connected by $k$ dephasing channels $\mathcal N^{(\text{deph})}_x$ and $1$ amplitude damping channel $\mathcal N^{(\text{ad})}_\lambda$, as shown in Fig.~\ref{fig: AD example}. 

\begin{figure}
	\centering
	\begin{tikzpicture}
	\draw (0,0) arc(-65:65:2);
	\draw (6,0) arc(180+65:180-65:2);
	\draw[->,decorate, decoration = snake] (1.5,1.67) -- (4.5,1.67);
	\draw[->] (1.2,2.68) -- (4.8, 2.68);
	\draw[->] (1.1,2.88) -- (4.9, 2.88);
	\draw[->] (1.2,0.68) -- (4.8, 0.68);
	\draw[->] (1.1,0.48) -- (4.9, 0.48);
	\node at (0,1.78) {\huge $A \mathcal C_A$};
	\node at (6,1.78) {\huge $\mathcal C_B B$};
	\node[above] at (3.1, 1.69) {\large $\mathcal N^{(\text{ad})}_\lambda$};
	\node[above] at (3.1, 2.9) {\large $\mathcal N^{(\text{deph})}_x \otimes \mathcal N^{(\text{deph})}_x$};
	\node[above] at (3.1, 0.68) {\large $\mathcal N^{(\text{deph})}_x \otimes \mathcal N^{(\text{deph})}_x$};
	\end{tikzpicture}
	\caption{Example of bipartition $A\mathcal C_A:\mathcal C_BB$ connected by $k = 4$ dephasing channels (straight lines) and $1$ amplitude damping channel (wiggling line). Once a bipartition of the network has been selected, it is not necessary to keep track of the precise nodes connected by the channels in order to apply Theorem \ref{thm: network bound}. \label{fig: AD example}}
\end{figure}
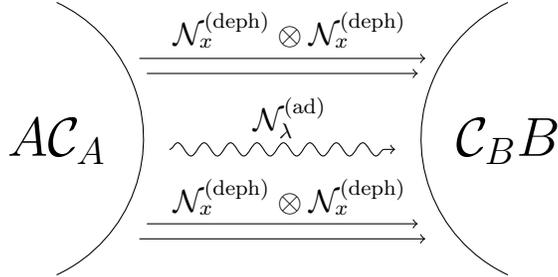

If we assume that all channels are used the same average number of times, we can express $\tilde \mu_{\mathcal C_A}$ as a function of $k$ and of the parameters $x,\lambda \in [0,1]$. In particular, we can write
\begin{equation}
\tilde \mu_{\mathcal C_A} = \frac{k[\tilde E_{\text{sq}}(\mathcal N^{(\text{deph})}_x) - E_{\text{R}}(\mathcal N^{(\text{deph})}_x)] + [\tilde E_{\text{sq}}(\mathcal N^{(\text{ad})}_\lambda) -  E_{\text{max}}(\mathcal N^{(\text{ad})}_\lambda)]}
{k[\tilde E_{\text{sq}}(\mathcal N^{(\text{deph})}_x) + E_{\text{R}}(\mathcal N^{(\text{deph})}_x)] + [\tilde E_{\text{sq}}(\mathcal N^{(\text{ad})}_\lambda) +  E_{\text{max}}(\mathcal N^{(\text{ad})}_\lambda)]},
\end{equation}
where $\tilde E_{\text{sq}}(\mathcal N^{(\text{deph})}_x)$ and  $\tilde E_{\text{sq}}(\mathcal N^{(\text{ad})}_\lambda)$ are respectively the best known upper bounds on $E_{\text{sq}}(\mathcal N^{(\text{deph})}_x)$ \cite{TGW_IEEE_14} and $E_{\text{sq}}(\mathcal N^{(\text{ad})}_\lambda)$ \cite{Pirandola_15}, which have been plotted as red dot-dashed curves in Figs. \ref{fig: dephasing} and \ref{fig: amp damp bounds}:
\begin{align}
\tilde E_{\text{sq}}(\mathcal N^{(\text{deph})}_x) &= h\left(\sqrt{\frac{x}{2}  \left(1-\frac{x}{2}\right)}+\frac{1}{2}\right),\\
\tilde E_{\text{sq}}(\mathcal N^{(\text{ad})}_\lambda) &= h\left(\frac{1}{2} - \frac{\lambda}{4}\right) - h\left(1 - \frac{\lambda}{4}\right),
\end{align}
where $h(y) \equiv - y \log_2 y - (1-y) \log_2 (1-y)$. 
Moreover, the quantity $E_{\text{max}}(\mathcal N^{(\text{ad})}_\lambda)$ has been shown to coincide with the upper bound obtained in Proposition~\ref{prop: bounds on Emax}, whereas the quantity $E_{\text{R}}(\mathcal N^{(\text{deph})}_x)$ is known to be equal to ${1- h(x/2)}$ \cite{Pirandola_15}.
The results obtained for $\tilde \mu_{\mathcal C_A} $ are plotted in Fig.~\ref{fig: examples} for $k = 1, 5, 10$ and $50$. As expected, we can see that the region of parameters $(x,\lambda)$ with $\tilde \mu_{\mathcal C_A} >0$, i.e., in which our versatile bound is advantageous, becomes larger with $k$. 
However, even for $k=1$ there is a broad set of parameters for which our versatile bound is tighter than the bound based on the squashed entanglement. In particular, this is the case when $\lambda \simeq 1$, because the negative contribution in  $\tilde \mu_{\mathcal C_A}$ from $ E_{\text{max}}(\mathcal N^{(\text{ad})}_\lambda) \geq \tilde E_{\text{sq}}(\mathcal N^{(\text{ad})}_\lambda)$ is close to zero. On the contrary, the bound based on the squashed entanglement is preferable when $x \simeq 1$, because 
$E_{\text{sq}}(\mathcal N^{(\text{deph})}_x)$ is close to zero and $E_{\text{R}}(\mathcal N^{(\text{deph})}_x)$ cannot be significantly smaller. The peak that can be observed in $\tilde \mu_{\mathcal C_A}$ for $x,\lambda \to 1$ is due to the fact that the upper bounds on the number of ebits (pbits) produced by the network go to zero, and small differences of one bound with respect to the other become significant.

\begin{figure}
	\centering
	\subfloat[$k = 1$.]{\includegraphics[scale = 0.63]{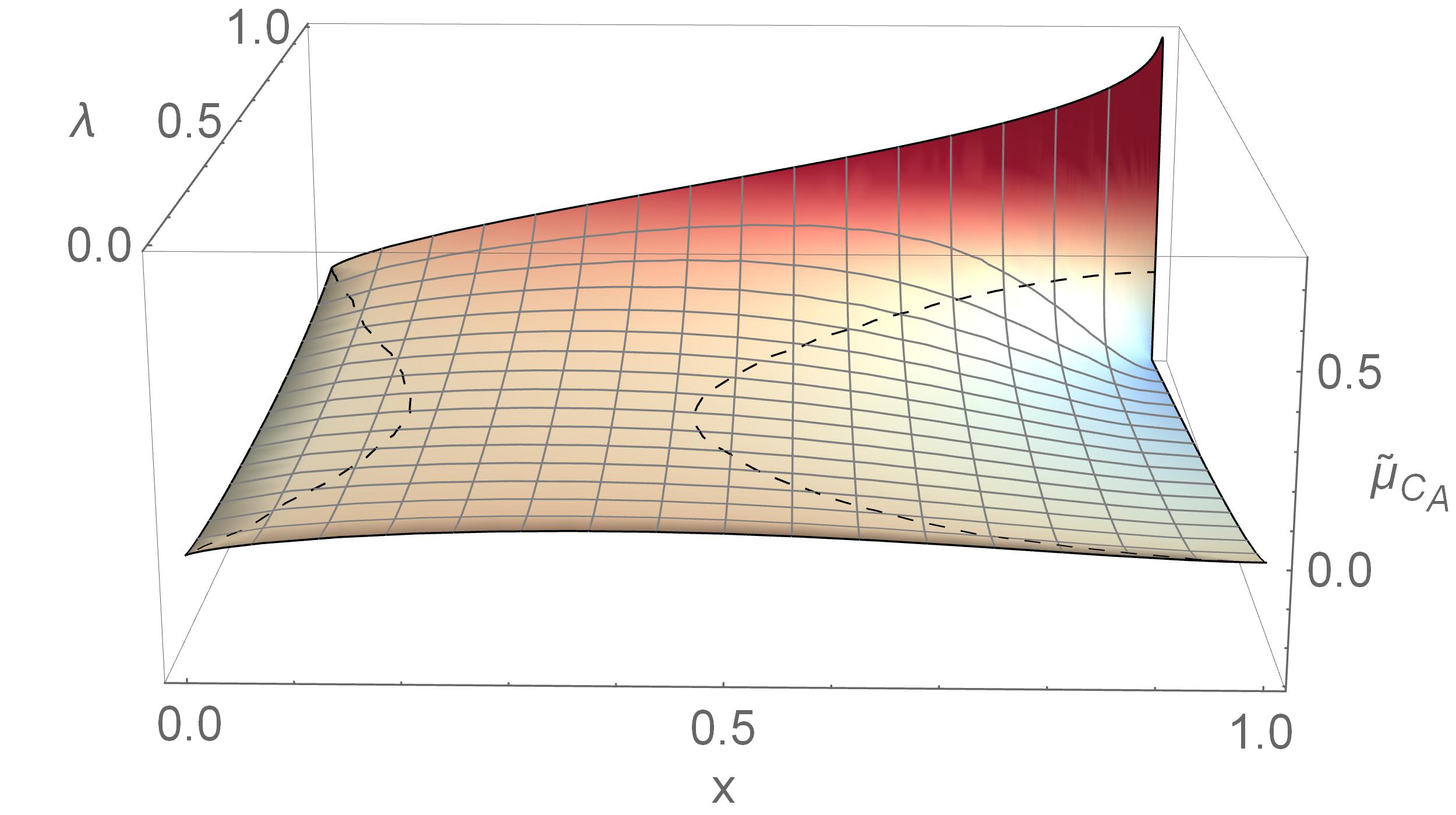}} 
	\subfloat[$k = 5$.]{\includegraphics[scale = 0.63]{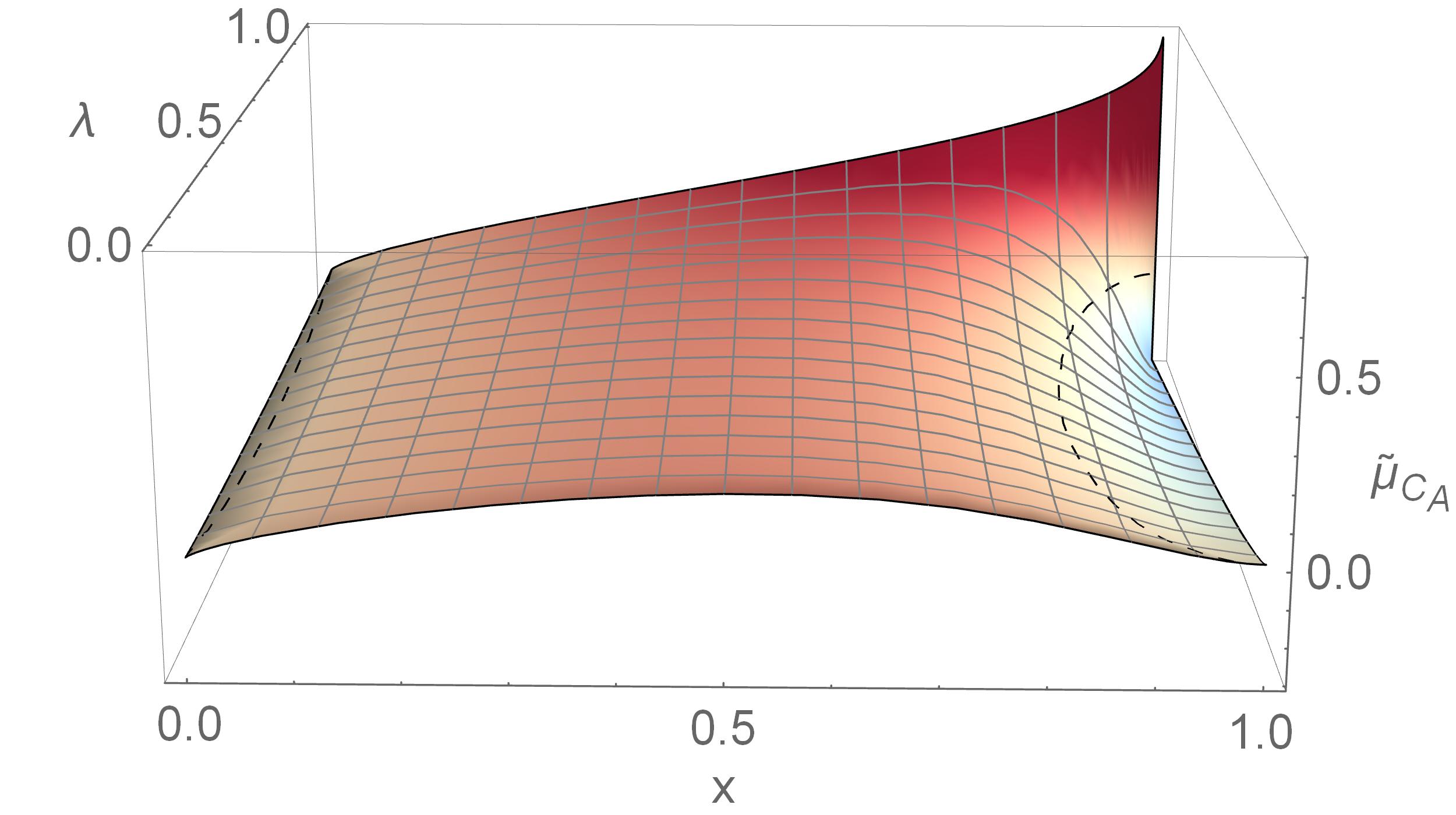}}\\
	\subfloat{\includegraphics[scale = 0.63]{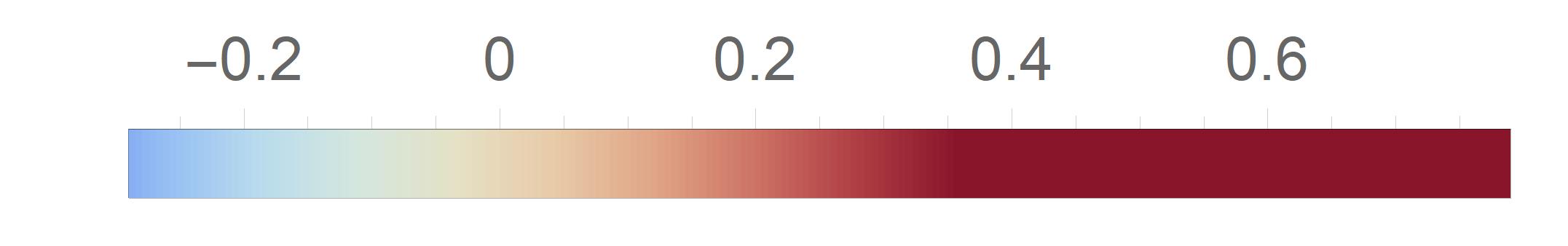}}\\
	\addtocounter{subfigure}{-1}
	\subfloat[$k = 10$.]{\includegraphics[scale = 0.63]{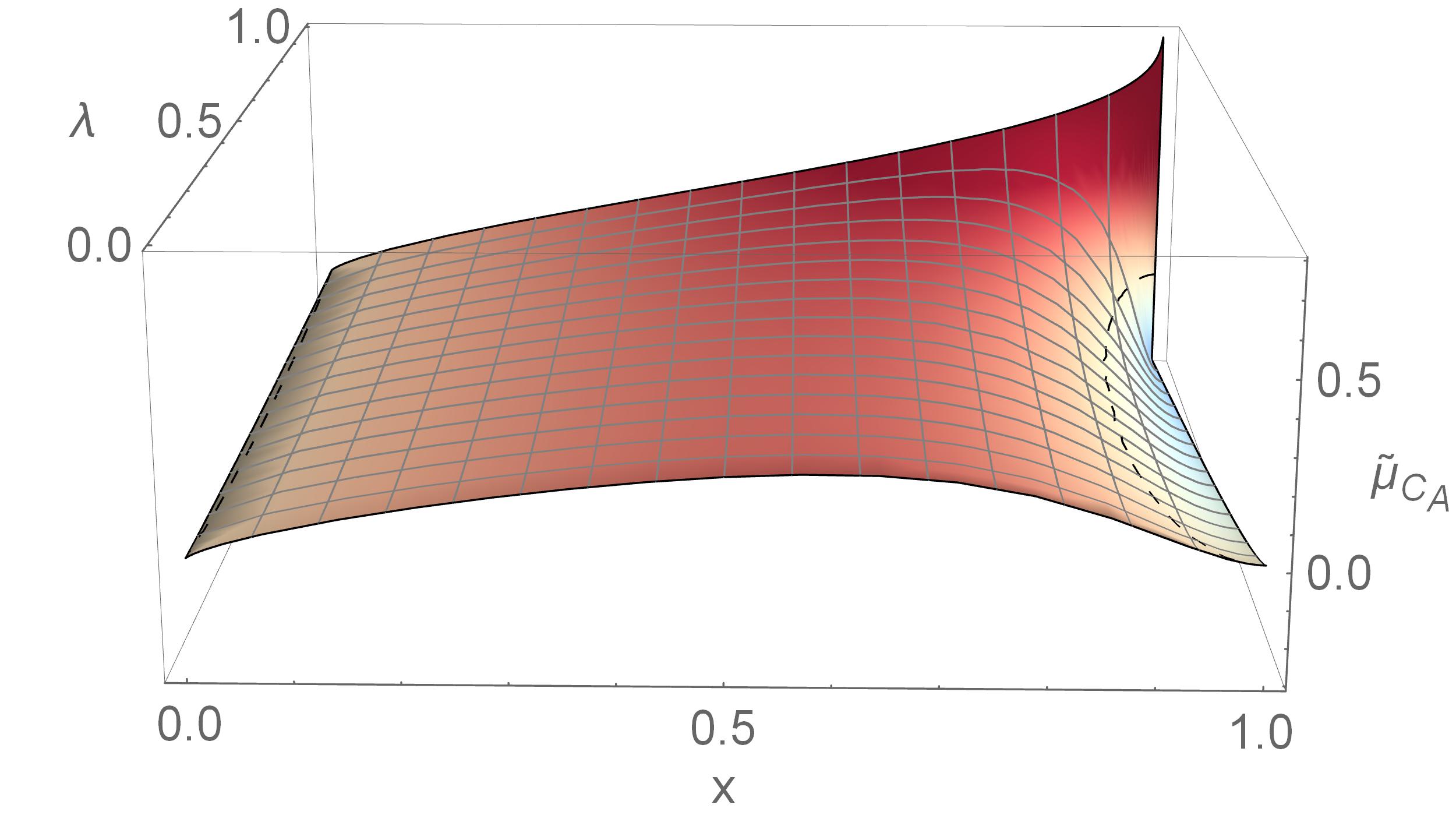}}
	\subfloat[$k = 50$.]{\includegraphics[scale = 0.63]{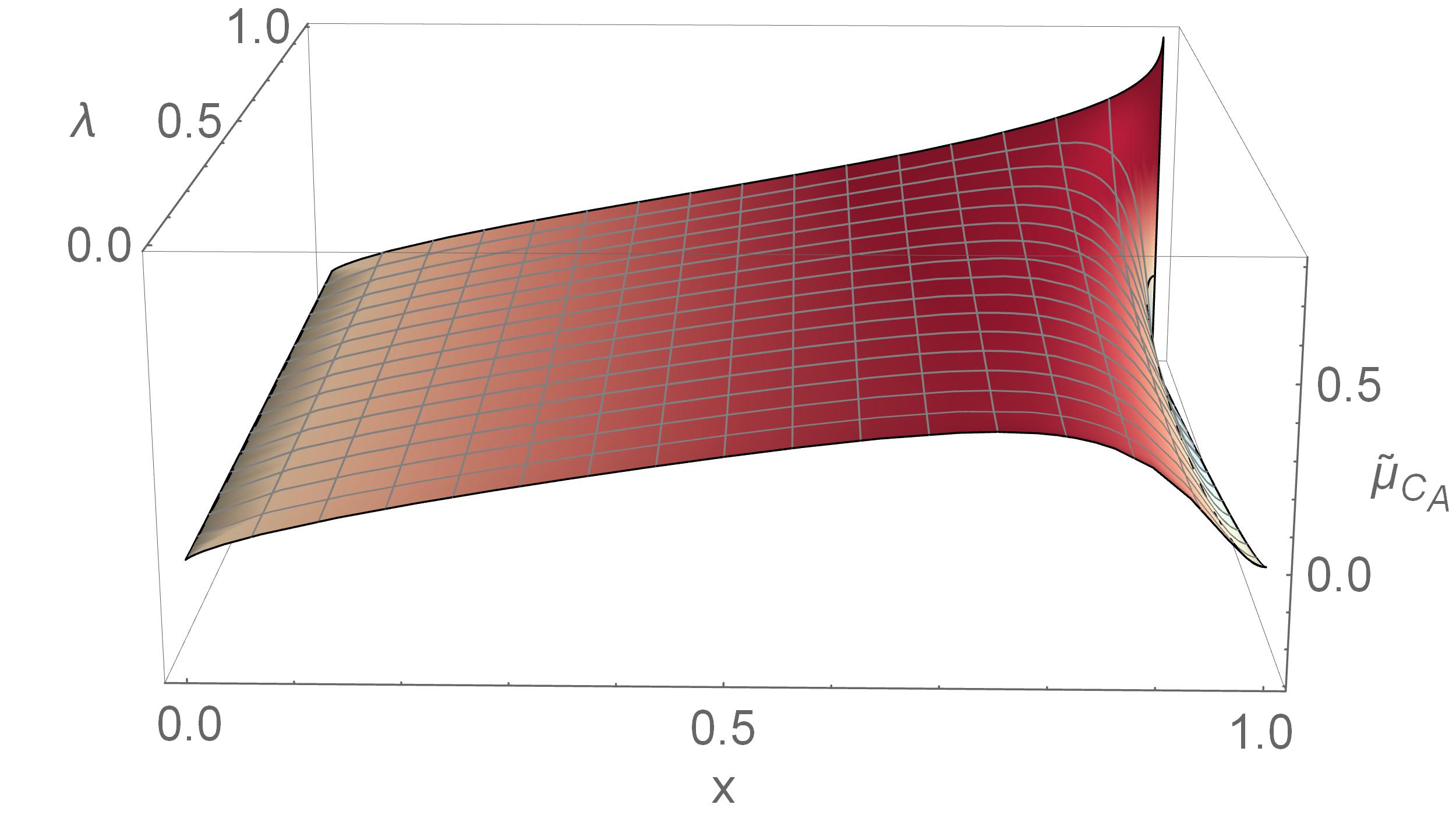}}
	\caption{Relative advantage of the versatile upper bound ${\mathcal E}^\prime$ over the upper bound ${\mathcal E}_{E_{\text{sq}}}$ based on the squashed entanglement, as measured by the parameter $\tilde \mu_{\mathcal C_A}$, for a bipartition of the network whose components are connected by $k$ dephasing channels $\mathcal N^{(\text{deph})}_x$ and $1$ amplitude damping channel $\mathcal N^{(\text{ad})}_\lambda$. The set of points characterised by $\tilde \mu_{\mathcal C_A} = 0$ is highlighted on the plots by dashed black curves. Our versatile bound is tighter than the best known upper bound based on the squashed entanglement on the regions where $\tilde \mu_{\mathcal C_A} > 0$.
		\label{fig: examples}
	}
\end{figure}

\section{Discussion and Conclusions}\label{sec: disc and conclusions}

In this paper, we investigated the possibility of using multiple entanglement measures in order to upper bound the number of ebits (or pbits) shared by two parties at the end of a communication protocol over a quantum network, with no limit on their classical communication. In particular, we exploited the special relation between the relative entropy and the max-relative entropy of entanglement, summarised by \Eq{eq: Emax bound}, in order to jointly use them in a single bound, which retains the advantages of both measures. For instance, it is possible to take advantage from the presence of Choi-simulable channels in the network, without requiring this property beforehand. 
From a theoretical perspective, our versatile bound performs much better than the previously known bound, which was based on the squashed entanglement, on networks composed by flower channels and Choi-simulable channels with $E_{\text{R}}$ smaller than $E_{\text{sq}}$. For more physically relevant quantum networks, in general one should check on a case-by-case basis which upper bound yields the tightest result. However, we can expect the versatile bound introduced in Theorem~\ref{thm: network bound} to be the best choice when the number of Choi-simulable channels is larger than the number of channels not satisfying this property, at least as long as $E_{\text{R}}$ provides tighter bounds than $E_{\text{sq}}$ on the Choi-simulable components of the network. This intuition was confirmed for a network composed by $k$ dephasing channels and one amplitude damping channel, where already for $k=5$ our versatile bound performed better on a broad range of parameters.

We should also reiterate that, according to the authors of Ref. \cite{Christandl_16}, \Eq{eq: generic Emax bound} has been rigorously proven only for channels acting on finite dimensional systems. As Theorem~\ref{thm: network bound} heavily relies upon that inequality, one should pay special attention when applying Theorem \ref{thm: network bound} to infinite dimensional channels, as long as the proof of \Eq{eq: generic Emax bound}  will not be suitably extended. Notice, however, that at least some bosonic channels (e.g., photon losses) are Choi-simulable: in these cases we can safely upper bound the entanglement of their output state via \Eq{eq: ER bound} \cite{Pirandola_15} and Theorem \ref{thm: network bound} still holds.

The advantage provided by our method would be further increased if more entanglement measures could be included within the same framework. An obvious candidate would be the squashed entanglement, because it typically provides tighter upper bounds on the capacity of a quantum channel than $E_{\text{max}}$, while being at the same time broadly applicable. This  research line could go together with the search for other entanglement measures that can provide upper bounds on channel capacities. From this point of view, we feel that the schematic framework provided by Theorems~\ref{thm: key hypothesis} and \ref{thm: ent-based network bound} could act as a guideline for future investigations. It would also be interesting to look into the possibility of extending this ``versatile'' approach to a multi-user scenario, where the network is composed by broadcast quantum channels \cite{Yard_11,Seshadreesan_16_broadcast,Baeuml_16,Laurenza_broadcast_2016,Takeoka_broadcast_2016,Takeoka_broadcast_17}.  

As a final remark, notice that the idea behind our result can be applied more generally in order to bound the rate at which a parallel composition of quantum channels can generate ebits (or pbits), when assisted by unlimited classical communication. 
Furthermore, although this paper has been developed from the perspective of quantum communication, it is worth stressing that the problem of quantifying the amount of bipartite, or multipartite, entanglement shared among the nodes of a network is also relevant from the perspective of quantum computation. In this paradigm, the quantum channels can be interpreted as noisy physical operations, and the nodes could represent, for example, the components of a cluster state. As the possibility of performing measurement-based universal quantum computation strongly depends on the entanglement of the initial resource state \cite{Briegel_MeasComp}, the ideas developed in this paper could also help in assessing the quality of entangled resources \cite{Yokoyama_cluster}, by considering $\mathcal P_{\epsilon,n}$ as the sequence of operations generating them.

\section*{Acknowledgements}
The authors thank Alexander M{\"u}ller-Hermes for discussions. We are also indebted with an anonymous referee for bringing to our attention the SDP formulation of the max-relative entropy of entanglement of qubit channels, as well as the sufficient conditions for the strong converse property mentioned in Corollary \ref{cor: capacity bound}.
LR expresses his gratitude to the Theoretical Quantum Physics Research Group of NTT BRL for the warm hospitality received during his visit, and acknowledges financial support from the People Programme (Marie Curie Actions) of the European Union’s Seventh Framework Programme (FP7/2007-2013) under REA Grant Agreement 317232. KA and GK thank support from the ImPACT Program of Council for Science, Technology and Innovation (Cabinet Office, Government of Japan). MSK acknowledges the UK EPSRC grant (EP/K034480/1), Samsung GRO programme and the Royal Society. WJM acknowledges support from the John Templeton Foundation (JTF \#60478). The opinions expressed in this publication are those of the author(s) and do not necessarily reflect the views of the John Templeton Foundation.

\newpage
\bibliography{Network}
\bibliographystyle{ieeetranmod2}

\pagebreak
\appendix

\section{Proof of Proposition~\ref{prop: bounds on Emax}}\label{app: proposition proofs}
	The lower bound is simply obtained by using the maximally entangled state $\psi_{A A^\prime}$ as input in \Eq{eq: Emax explicit}, without optimising over all $\rho_{A A^\prime}$. Furthermore, its equality with $E_{\text{max}} (\mathcal N)$ itself in the case of Choi-simulable channels can be obtained as in the last step of \Eq{eq: ER bound}. Indeed, that argument holds for any entanglement measure and not only for $E_{\text{R}}$.
	
	The upper bound, on the other hand, is a re-elaborated version of the upper bound on the max-relative entropy of a channel studied in Ref. \cite{Christandl_16}. In order to obtain their result, the authors introduce a generic entanglement breaking (EB) channel $\mathcal T_{A^\prime\to B}$, and use the following chain of inequalities: 
	\begin{align}
	E_{\text{max}}(\mathcal N) &
	\leq \max_{\rho_{A A^\prime}} \min_{\mathcal T_{A^\prime\to B} \in \text{EB}} \inf_x\{x\in\mathbb{R} \vert \left(2^x \mathcal T_{A^\prime\to B} - \mathcal N_{A^\prime\to B}\right)[\rho_{AA^\prime}]\geq 0 \} \notag \\
	&\leq \min_{\mathcal T \in \text{EB}} \inf_x\{x\in\mathbb{R} \vert 2^x \pi_{\mathcal T} - \pi_{\mathcal N}\geq 0 \} 
	= \min_{\mathcal T \in \text{EB}}  D_{\text{max}}(\pi_{\mathcal N}\Vert\pi_{\mathcal T}). \label{eq: first UB on Emax}
	\end{align}
	The first inequality is obtained by optimising over a smaller set of separable states, in which $\sigma_{AB}$ is obtained as output of entanglement-breaking channels acting on the same input state $\rho_{A A^\prime}$. The second inequality is then obtained by noticing that $\left(2^x \mathcal T_{A^\prime\to B} - \mathcal N_{A^\prime\to B}\right)[\rho_{AA^\prime}]\geq 0$ for any input $\rho_{A A^\prime}$ if the operator $\left(2^x \mathcal T_{A^\prime\to B} - \mathcal N_{A^\prime\to B}\right)$ is completely positive, and that this last condition is implied by the positivity of its Choi-Jamio\l{}kowski state. 
	In order to obtain the upper bound of Proposition~\ref{prop: bounds on Emax}, we just need to show that the set of  states $\pi_{\mathcal T}$ appearing in \Eq{eq: first UB on Emax} corresponds to the set of separable density matrices $\sigma_{AB}$ such that $\TrS{B}{\sigma_{AB}} = \Id_A/d$. One inclusion is trivial, while the other follows from the fact that, for any such $\sigma_{AB}$, we can find a corresponding completely positive and trace preserving (CPTP) map $T^{(\sigma_{AB})} \in \rm{EB}$ via the teleportation protocol:
	\begin{equation}\label{eq: channel map from choi state}
	\mathcal T_{A^\prime\to B}^{(\sigma_{AB})} (\tau_{A^\prime}) = d^2 \,\TrS{A^\prime A}{\psi_{A^\prime A} \left(\tau_{A^\prime} \otimes\sigma_{AB}\right)},
	\end{equation}
	where $\psi_{A^\prime A}$ is a maximally entangled state. Indeed, this map is CPTP because from \Eq{eq: channel map from choi state} we obtain a possible set of Kraus operators given by:
	\begin{equation}
	N_{A^\prime\to B}^{(h,k)} = d \, {}_{A^\prime A}\hspace*{-0.1cm}\bra{\psi}\sqrt{\sigma_{AB}}\ket{k}_A\ket{h}_B,
	\end{equation}
	with $\sum_{h,k=1}^d \left(N_{A^\prime\to B}^{(h,k)}\right)^\dagger N_{A^\prime\to B}^{(h,k)} = \Id_{A^\prime}$, and a straightforward calculation shows that $\pi_{\mathcal T^{(\sigma_{AB})}} = \sigma_{AB}$, thus proving that $\mathcal T^{(\sigma_{AB})} \in \rm{EB}$ because of the separability of $\sigma_{AB}$.

\section{Bounding the max-relative entropy of entanglement of qubit channels invariant under phase rotations}\label{app: tech tools}
Most of the typical qubit channels are invariant under rotations around the axis associated with the Pauli matrix $\sigma_z = \text{Diag}(+1,-1)$, and it is thus interesting to study the consequences of this fact for the evaluation of the upper and lower bounds identified in Proposition~\ref{prop: bounds on Emax}.
Let $\mathcal N$ be a quantum channel acting on a qubit, such that 
\begin{equation}\label{def: phase rotation}
\mathcal N(e^{i \theta \sigma_z} \rho e^{-i \theta \sigma_z}) = e^{i \theta \sigma_z}\mathcal N(\rho)  e^{-i \theta \sigma_z},
\end{equation}
for all angles $\theta$ and input states $\rho$. As the maximally entangled state $\psi_{A A^\prime}$ is left invariant by the unitary operation
\begin{equation}\label{eq: Utheta structure}
U_\theta = e^{+i \frac{\theta}{2} \sigma^{(A)}_z }\otimes e^{-i \frac{\theta}{2} \sigma^{(B)}_z},
\end{equation}
we can conclude that its Choi state $\pi_{\mathcal N}$ is also invariant under $U_\theta$, for any $\theta \in [0,2\pi]$.
This immediately implies that the average of $\pi_{\mathcal N}$ over all possible $\theta$ rotations coincides with $\pi_{\mathcal N}$ itself:
\begin{equation} \label{eq: gamma average invariance}
\pi_{\mathcal N} = \int \frac{\text{d}\theta}{2\pi} U_\theta \pi_{\mathcal N} U_\theta^\dagger.
\end{equation}
This allows us to prove the following lemma, whose proof can be found at the end of this appendix.

\begin{lem}\label{lem: optimisation on smaller set}
	Let $\pi_{\mathcal N}$ be a bipartite state invariant under the separable unitary evolution $U_\theta$ defined in \Eq{eq: Utheta structure}, and $\sigma_{AB}^*$ be the state which minimises $D_{\rm{max}}(\pi_{\mathcal N}\Vert\sigma_{AB})$ among all separable states $\sigma_{AB}$. If $\overline{\sigma^*}_{AB}$ is the averaged version of $\sigma^*_{AB}$,
	\begin{equation}
	\overline{\sigma^*}_{AB} \equiv \int \frac{\text{d}\theta}{2\pi} U_\theta \sigma_{AB}^* U_\theta^\dagger,
	\end{equation}
	then $\overline{\sigma^*}_{AB}$ is separable and
	\begin{equation}
	D_{\rm{max}}\big(\pi_{\mathcal N}\big\Vert\sigma^*_{AB}\,\big) =
	D_{\rm{max}}\Big(\pi_{\mathcal N}\big\Vert\overline{\sigma^*}_{AB}\,\Big).
	\end{equation}
	Similarly, if $\sigma_{AB}^*$ is the state which minimises $D_{\rm{max}}(\pi_{\mathcal N}\Vert\sigma_{AB})$ over all separable states $\sigma_{AB}$ with ${\rm{Tr}}_{B}[\sigma_{AB}] = \Id_A/2$, the same conclusion holds with ${\rm{Tr}}_{B}[\overline{\sigma^*}_{AB}]= \Id_A/2$. 
\end{lem}

As a corollary of Lemma \ref{lem: optimisation on smaller set}, we can restrict the minimisation over all separable states $\sigma_{AB}$ in \Eq{eq: bounds on Emax} to be only over the states which are left unaltered by being averaged over all possible $\theta$ rotations. The density matrix associated with these states in the basis $\{\ket{00},\ket{01},\ket{10},\ket{11}\}$
has the form 
\begin{equation}\label{eq: sigma generic form}
\sigma_{AB} = \frac{1}{2}\left(\begin{array}{c|cc|c}
\alpha & & & \xi e^{i\phi}  \\\hline
& \gamma & & \\
& & \delta & \\\hline
\xi e^{-i\phi} & & & \beta
\end{array}\right),
\end{equation}
with $\alpha,\beta,\gamma,\delta,\xi \geq 0$, $\alpha + \beta +  \gamma + \delta = 2$, $\phi \in [0,2\pi]$ and $0\leq\xi \leq \min\{\sqrt{\alpha\beta},\sqrt{\gamma\delta}\}$.
Note that the last inequality comes from the PPT criterion, which works for two-qubit states as a necessary and sufficient condition for separability \cite{Horodecki_PPT_96}. 
When evaluating the upper bound in Proposition~\ref{prop: bounds on Emax}, we simply need to add the additional constraints $\gamma = 1-\alpha$ and $\delta = 1-\beta$, in order to assure $\TrS{B}{\sigma_{AB}} = \Id_A/2$.

\begin{proof}[Proof of Lemma \ref{lem: optimisation on smaller set}]
	The max-relative entropy $ D_{\text{max}}(\rho\Vert\sigma)$ is invariant under joint unitary operations applied on both $\rho$ and $\sigma$, and is jointly quasi-convex. Both these properties have been previously introduced in Sec.~\ref{sec: entropy defs}, respectively in \Eq{def: invariance joint unitaries} and \Eq{def: joint quasi convexity}.
	Together with \Eq{eq: gamma average invariance}, these facts lead to 
	\begin{align}
	D_{\text{max}}\Big(\pi_{\mathcal N^{(\text{ad})}_\lambda}\Big\Vert\overline{\sigma^*}_{AB}\,\Big) & = 	
	D_{\text{max}}\bigg(\int \frac{\text{d}\theta}{2\pi} U_\theta \pi_{\mathcal N^{(\text{ad})}_\lambda} U_\theta^\dagger
	\bigg\Vert
	\int \frac{\text{d}\theta}{2\pi} U_\theta \sigma_{AB}^* U_\theta^\dagger\bigg) \notag\\
	& 	\leq \max_\theta  D_{\text{max}}( U_\theta \pi_{\mathcal N^{(\text{ad})}_\lambda} U_\theta^\dagger\Vert U_\theta \sigma_{AB}^* U_\theta^\dagger) = 
	D_{\text{max}}( \pi_{\mathcal N^{(\text{ad})}_\lambda} \Vert  \sigma_{AB}^* ).
	\end{align}
	The converse inequality follows because $\overline{\sigma^*}_{AB}$ is separable, due to the structure of $U_{\theta}$ [see \Eq{eq: Utheta structure}], and  because $\sigma_{AB}^*$ minimises $D_{\text{max}}(\pi_{\mathcal N}\Vert\sigma_{AB})$ over all separable states. The final remark can be easily proven by noticing that $\TrS{B}{\overline{\sigma^*}_{AB}}=\Id_A/2$ if $\TrS{B}{\sigma^*_{AB}}=\Id_A/2$.
\end{proof}

\section{Proof for the upper bound in \Eq{eq: explicit bounds on amplitude damping}} \label{app: upper bound amp damp}
	In order to prove the desired result, we need to explicitly perform the optimisation appearing in the upper bound of Proposition~\ref{prop: bounds on Emax}, i.e.
	\begin{equation}
	\tilde E_{\text{max}}(\mathcal N^{(\text{ad})}_\lambda) \equiv  \min_{\substack{\sigma_{AB}\in \rm{SEP} \\ {\rm{Tr}}_{B}[\sigma_{AB}]=\Id_A/2}}  D_{\rm{max}}(\pi_{\mathcal N^{(\text{ad})}_\lambda}\big\Vert\sigma_{AB})
	= \min_{\sigma_{AB}} \inf \{x\in\mathbb{R} \vert 2^x \sigma_{AB} - \pi_{\mathcal N^{(\text{ad})}_\lambda} \geq 0\},
	\end{equation} 
	where thanks to Lemma \ref{lem: optimisation on smaller set} on the rightmost term we can consider only states $\sigma_{AB}$ with the form given in \Eq{eq: sigma generic form}, with $\gamma = 1-\alpha$ and $\delta = 1-\beta$. Let us introduce the parameter $y = 2^x$. By explicitly computing the Choi-Jamio\l{}kowski state $\pi_{\mathcal N^{(\text{ad})}_\lambda}$, the condition $y \sigma_{AB} - \pi_{\mathcal N^{(\text{ad})}_\lambda} \geq 0$ can be rewritten as the system of inequalities:
	\begin{equation}\label{eq: conditions}
	\begin{cases}
	y(1-\beta)\geq {\lambda}, \\
	y \tilde\sigma - \tilde \pi_\lambda \geq 0,
	\end{cases}
	\end{equation}
	where $\tilde \sigma$ and $\tilde \pi$ are $2\times 2$ matrices
	\begin{equation}
	\tilde{\sigma} = \left(\begin{array}{cc}
	\alpha & \xi e^{i\phi} \\
	\xi e^{-i\phi} & \beta
	\end{array}\right), \qquad
	\tilde{\pi}_\lambda = \left(\begin{array}{cc}
	1 & \sqrt{1-\lambda} \\
	\sqrt{1-\lambda} & 1-\lambda
	\end{array}\right).
	\end{equation}
	We now define $y_1(\lambda,\sigma_{AB})$ and $y_2(\lambda,\sigma_{AB})$ as the smallest values of $y$ that satisfy respectively the first and the second inequalities appearing in \Eq{eq: conditions}, and we rewrite the minimisation leading to the upper bound on $E_{\text{max}}(\mathcal N^{(\text{ad})}_\lambda)$ as
	\begin{equation}
	 \tilde E_{\text{max}}(\mathcal N^{(\text{ad})}_\lambda) \equiv \min_{\sigma_{AB}} \inf \{x\in\mathbb{R} \vert 2^x \sigma_{AB} - \pi_{\mathcal N^{(\text{ad})}_\lambda} \geq 0\} = \log_2 \min_{\sigma_{AB}} \max \{y_1(\lambda,\sigma_{AB}),y_2(\lambda,\sigma_{AB})\}.
	\end{equation} 		
	We can easily show that this quantity is smaller than or equal to $\log_2(2-\lambda)$ by providing a matrix $\sigma_{AB}$ of the desired form such that $\max \{y_1(\lambda,\sigma_{AB}),y_2(\lambda,\sigma_{AB}) \}= 2-\lambda$. This can be achieved with the choices:
	\begin{equation}
	\alpha = \frac{1}{2-\lambda},\qquad \beta = 1-\alpha, \qquad\xi = \sqrt{\alpha\beta}, \qquad \phi=0,
	\end{equation}
	which yield $y_1 = \lambda(2-\lambda)$ and $y_2 = 2-\lambda$, as can be verified by directly substituting these values into \Eq{eq: conditions}.
	The converse inequality, i.e. $\tilde E_{\text{max}}(\mathcal N^{(\text{ad})}_\lambda) \geq \log_2(2-\lambda)$,  requires some additional work. Thanks to the monotonicity of the logarithm and the trivial relation $\max \{y_1,y_2\}\geq y_2$, we can bound $\tilde E_{\text{max}}(\mathcal N^{(\text{ad})}_\lambda) $ from below as
	\begin{equation}
\tilde E_{\text{max}}(\mathcal N^{(\text{ad})}_\lambda)  \geq  \log_2 \min_{\sigma_{AB}} y_2(\lambda,\sigma_{AB}).
	\end{equation}
	Hence, we are left with the task of showing that $\min_{\sigma_{AB}} y_2(\lambda,\sigma_{AB}) \geq 2 - \lambda$, where the optimisation has to be effectively performed over the parameters $\alpha,\beta,\xi,\phi$ satisfying the conditions detailed after \Eq{eq: sigma generic form}, with $\gamma = 1-\alpha$ and $\delta = 1-\beta$. 
	
	The condition $y \tilde\sigma - \tilde \pi_\lambda \geq 0$ involves $2 \times 2$ matrices, and can be rewritten using Pauli matrices $\vec{\sigma} = \{\sigma_x,\sigma_y,\sigma_z\}$ as
	\begin{equation}
	y(\alpha + \beta)(\Id + \vec{v}\cdot \vec{\sigma}) -(2-\lambda)(\Id + \hat{n}\cdot\vec{\sigma}) \geq 0, 
	\end{equation}
	where
	\begin{equation}
	\vec{v}  = \frac{1}{\alpha+\beta} \left(\begin{array}{c}
	2 \xi \cos\phi \\
	- 2 \xi \sin\phi\\
	\alpha - \beta
	\end{array}\right), \qquad 
	\hat{n}  = \frac{1}{2-\lambda} \left(\begin{array}{c}
	2\sqrt{1-\lambda} \\
	0\\
	\lambda
	\end{array}\right).
	\end{equation}
	This in turn reduces to
	\begin{equation} \label{eq: y2 condition}
	y \geq \frac {2-\lambda}{\alpha + \beta} \; \frac{2(1-v\cos\psi)}{1-v^2} \equiv y_2(\lambda,\sigma_{AB}),
	\end{equation}
	where $v = |\vec{v}|\leq 1$ and $\psi$ is the angle between $\vec{v}$ and $\hat{n}$. Note that the second fraction appearing in \Eq{eq: y2 condition} is always larger than $1$, therefore, when $\alpha+\beta \leq 1$ the condition $y_2(\lambda,\sigma_{AB}) \geq 2-\lambda$ holds. On the other hand, if $1\leq \alpha + \beta \leq 2$, we can use the parametrisation:
	\begin{equation}
	2\xi = \eta(2-\alpha-\beta) \sin\zeta, \qquad \alpha-\beta = \eta(2-\alpha-\beta) \cos\zeta,
	\end{equation}
	with $\eta \in [0,1]$ and $\zeta \in [0,\pi]$.
	This allows us to conclude because of the following chain of inequalities:
	\begin{align}
	y_2(\lambda,\sigma_{AB}) &= 2 (2-\lambda) \frac{(\alpha+\beta) - \eta(2-\alpha-\beta)\left[\cos(\theta-\zeta)-\sin\theta\sin\zeta(1-\cos\phi)\right]}{(\alpha+\beta)^2-\eta^2(2-\alpha-\beta)^2} \notag\\
	& \geq 2 (2-\lambda) \frac{(\alpha+\beta) - \eta(2-\alpha-\beta)}{(\alpha+\beta)^2-\eta^2(2-\alpha-\beta)^2} \notag \\
	& = (2-\lambda) \frac{2}{2\eta +(1-\eta)(\alpha+\beta)} \geq (2-\lambda),
	\end{align}
	where $\theta = \arctan(2\sqrt{1-\lambda}/\lambda)$ is the angle describing the direction of $\hat{n}$.

\section{Proof for the lower bound in \Eq{eq: explicit bounds on amplitude damping}}\label{app: lower bound proof}
The goal of this Appendix is to provide a proof for the following lower bound on $E_{\text{max}}(\mathcal N^{(\text{ad})}_\lambda)$:
\begin{equation}
E_{\text{max}}(\mathcal N^{(\text{ad})}_\lambda) \geq \min_{\sigma_{AB}\in \text{SEP}} \tilde D_{\text{max}}(\pi_{\mathcal N^{(\text{ad})}_\lambda}\Vert\sigma_{AB}) = 
\begin{cases}
\log_2\left(\frac{1}{2}(1+\sqrt{1-\lambda})^2\right), & \text{if }\lambda\leq \frac{\sqrt{5}-1}{2},\\
\log_2\left(\frac{1+\lambda}{2\lambda}\right), &\text{if } \lambda \geq \frac{\sqrt{5}-1}{2}.
\end{cases}
\end{equation}
Thanks to Lemma \ref{lem: optimisation on smaller set}, we can reduce the optimisation over all separable states $\sigma_{AB}$ that are left unaltered under all possible $\theta$ rotations, which can be parametrised as in \Eq{eq: sigma generic form}. The condition $y \sigma_{AB}  - \pi_{\mathcal N^{(\text{ad})}_\lambda} \geq 0 $ can be explicitly rewritten as
\begin{equation}
\begin{cases}
y \geq \frac{\lambda}{\delta}, \\
y \geq \frac{\alpha(1-\lambda) + \beta -2\sqrt{1-\lambda}\xi\cos\phi}{\alpha\beta - \xi^2},
\end{cases}
\end{equation}
so that
\begin{equation}
\tilde D_{\text{max}}\left(\pi_{\mathcal N^{(\text{ad})}_\lambda}\big\Vert\sigma_{AB}\right) = \log_2\max\left\{\frac{\alpha(1-\lambda) + \beta -2\xi\cos\phi\sqrt{1-\lambda}}{\alpha\beta - \xi^2},\frac{\lambda}{\delta} \right\}.
\end{equation}
In what follows, for any fixed $\lambda$ we will minimise this quantity over the parameters $\alpha,\beta,\gamma,\delta,\xi,\phi$, satisfying the constraints detailed after \Eq{eq: sigma generic form}.

The minimisation in $\phi$ can be easily performed, with the optimal choice being $\phi=0$. Moreover, for any fixed $\alpha,\beta,\xi$, the maximum $\delta$ (and thus the minimum $\lambda/\delta$) is given by 
\begin{equation}
\delta_{\text{max}}=\frac{1}{2}\left(2-\alpha - \beta + \sqrt{(2-\alpha - \beta)^2 - 4\xi^2}\right),
\end{equation} 
that is, when  $\delta>\gamma$ and $\gamma\delta$ equals the smallest allowed value $\xi^2$.
Notice that this choice implies 
\begin{equation}
2\xi = 2\sqrt{\gamma\delta} \leq \gamma + \delta = 2 - \alpha - \beta.
\end{equation} 
At this stage, the optimisation problem (without the logarithm) has been reduced to:
\begin{equation}
\min_{\alpha,\beta,\xi} \left\{ \max\left[
\frac{\alpha(1-\lambda) + \beta -2\xi\sqrt{1-\lambda}}{\alpha\beta - \xi^2},
\frac{2 \lambda}{2-\alpha - \beta + \sqrt{(2-\alpha - \beta)^2 - 4\xi^2}}
\right] \; 
\Bigg| 
	\begin{aligned}
	\scriptstyle \alpha,\beta,\xi \geq 0 \; \wedge \; \xi^2 \leq \alpha\beta \\ 	\scriptstyle  \wedge \; \alpha + \beta + 2\xi \leq 2
	\end{aligned}	
\right\}.
\end{equation}
Now we introduce the parameters $\nu = (\alpha + \beta)/2$ and $\mu = (\alpha-\beta)/2$. 
As $\alpha > \delta$ always yields a smaller value than the converse choice, we can limit our study to $\mu\geq 0$ and rewrite the problem in the new parameters:
\begin{equation}
\min_{\nu,\mu,\xi} \left\{ \max\left[
\frac{\nu(2-\lambda) -\lambda\mu -2\xi\sqrt{1-\lambda}}{\nu^2 - \xi^2 - \mu^2},
\frac{\lambda}{1-\nu + \sqrt{(1-\nu)^2 - \xi^2}}
\right] \; 
\Bigg|  
\begin{aligned}
\scriptstyle 0 \leq \mu \leq \sqrt{\nu^2 - \xi^2} \\ 
\scriptstyle 0\leq \xi \leq \nu  \; \wedge \; \nu + \xi \leq 1
\end{aligned}
\right\}.
\end{equation}
We can now minimise the first term over $\mu$. The value $\mu_0$ for which the function 
\begin{equation}\label{eq: f0 function}
f_0(\mu | \lambda,\nu,\xi) = \frac{\nu(2-\lambda) -\lambda\mu -2\xi\sqrt{1-\lambda}}{\nu^2 - \xi^2 - \mu^2}
\end{equation}
becomes zero is always bigger than $\sqrt{\nu^2 - \xi^2}$ in the considered region. Together with the asymptotic scaling $f_0(\mu | \lambda,\nu,\xi) \sim \lambda/\mu$ for $|\mu|\gg 1$, this can be used to deduce the qualitative behaviour of $f_0(\mu | \lambda,\nu,\xi)$, which is shown in Figure \ref{fig: typical plot}. Let $\mu_{\pm}(\lambda,\nu,\xi)$ be the zeros of $\partial_\mu f_0(\mu | \lambda,\nu,\xi)$, with $\mu_- \leq \mu_+$, where
\begin{equation}
\mu_{\pm}(\lambda,\nu,\xi) = \frac{1}{\lambda}\left[(2-\lambda)\nu - 2\sqrt{1-\lambda}\xi\right]\pm\frac{1}{\lambda}|2\sqrt{1-\lambda}\nu-(2-\lambda)\xi|.
\end{equation}
As $f_0(\mu_-| \lambda,\nu,\xi) \geq f_0(\mu_+ | \lambda,\nu,\xi)$, we can find the desired minimum of $f_0(\mu| \lambda,\nu,\xi)$ in $\mu\in[0,\sqrt{\nu^2-\xi^2}]$ as 
\begin{equation}
\min_{\mu} f_0(\mu| \lambda,\nu,\xi) = \max \{f_0(\mu_-),f_0(\mu+)\} = \max\left\{\frac{(1-\sqrt{1-\lambda})^2}{2(\nu-\xi)},\frac{(1+\sqrt{1-\lambda})^2}{2(\nu+\xi)} \right\}.
\end{equation}

\begin{figure}
	\includegraphics[scale = 0.7]{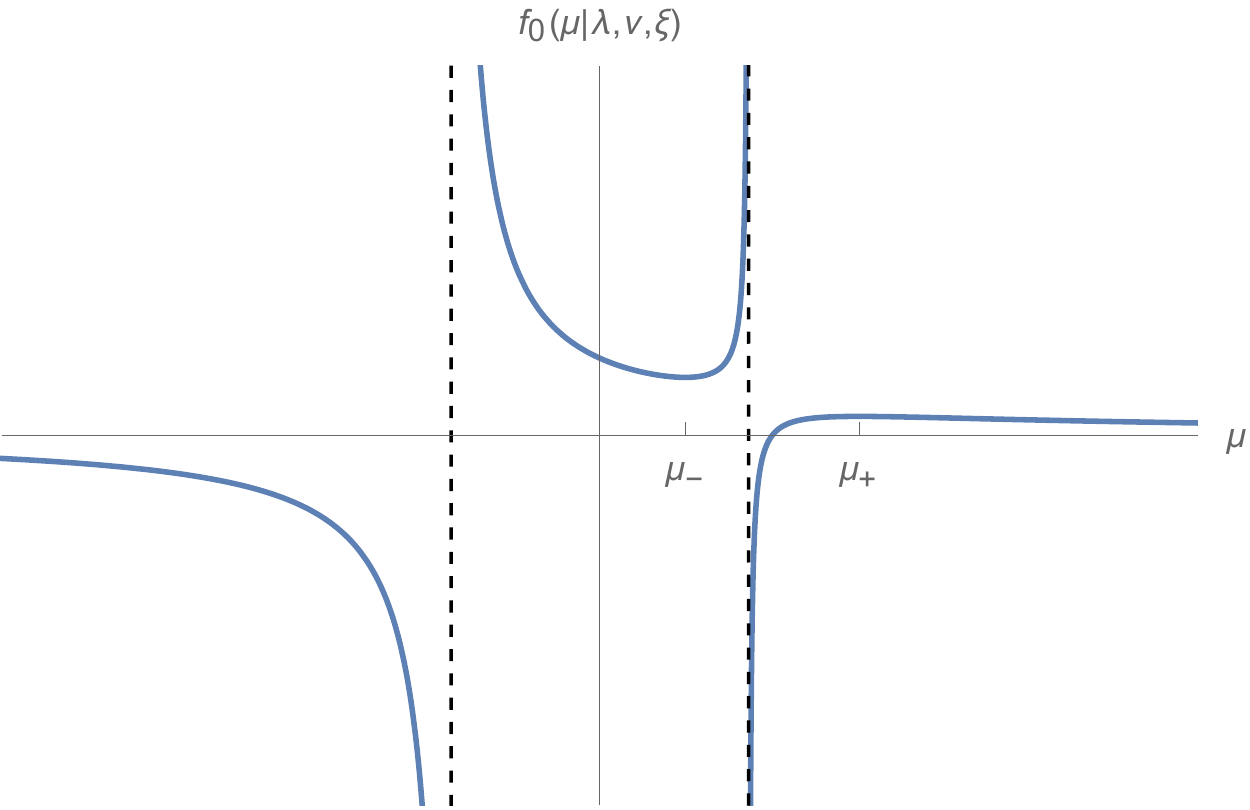}
	\caption{Typical plot of the function $f_0(\mu | \lambda,\nu,\xi) $ defined in \Eq{eq: f0 function}.
		\label{fig: typical plot}
	}
\end{figure}

It is worth substituting $\nu \to x(1+y)/2 $ and $\xi \to x(1-y)/2$. In terms of the new variables, the problem after the optimisation in $\mu$ becomes
\begin{equation}
\min_{x,y} \left\{ \max\left[
\frac{(1-\sqrt{1-\lambda})^2}{2xy},
\frac{(1+\sqrt{1-\lambda})^2}{2x},
\frac{\lambda}{1-x\frac{(1+y)}{2} + \sqrt{(1-x)(1-xy)}}
\right] \; 
\Bigg| 
\begin{aligned}
\scriptstyle 0 \leq x \leq 1 \\\scriptstyle 0 \leq y \leq 1
\end{aligned}
\right\},
\end{equation}
whose form is suitable to perform the minimisation in $y$. Let us label the three function appearing between square brackets in order as $f_1,f_2$ and $f_3$. Note that $f_1$ and $f_3$ are respectively monotonically decreasing and increasing with $y$, with only the first one diverging to infinity for $y\to 0$. If the two functions do not cross each other, i.e., if $x\leq x_{th} \equiv (1-\sqrt{1-\lambda})/2$, the minimum over $y$ is thus obtained by evaluating $f_1$ in $y=1$, otherwise we need to pick their intersection point. Explicitly, this can be written as
\begin{equation}
\min_{y}\{f_1,f_3\} = \begin{cases}
\frac{(1-\sqrt{1-\lambda})^2}{2 x}, & \text{if} \, 0\leq x \leq x_{\rm{th}}, \\
f_4(x,\lambda), & \text{if} \, x_{\rm{th}} \leq x \leq 1,
\end{cases}
\end{equation}
where
\begin{equation}
f_4(x,\lambda) = \frac{1}{2 x^2}\left\{8+x\left[(\frac{1-\sqrt{1-\lambda}}{\sqrt{\lambda}})^2-4\right] -4 \sqrt{(1-x)\left[4+x\left(\frac{1-\sqrt{1-\lambda}}{\sqrt{\lambda}}\right)^2\right]}\right\}.
\end{equation}
Finally, we can optimise over $x$. If $x\leq x_{\rm{th}}$, we are left with:
\begin{equation}\label{eq: temp minimum}
\min_{x\leq x_{\rm{th}}} \max \left\{\frac{(1-\sqrt{1-\lambda})^2}{2 x},\frac{(1+\sqrt{1-\lambda})^2}{2 x} \right\} = \frac{(1+\sqrt{1-\lambda})^2}{2 x_{\rm{th}}} = f_2(x_{\rm{th}},\lambda).
\end{equation}
On the other hand, when $x\geq x_{\rm{th}}$, we can apply the same reasoning used for the minimisation over $y$. In particular, $f_2$ and $f_4$ are respectively monotonically decreasing and increasing with $x$, $f_2(x_{\rm{th}})\geq f_4(x_{\rm{th}})$, and they have a crossing point only when $\lambda \geq (\sqrt{5}-1)/2$. If there is no crossing, the minimum over $x$ is given by $f_2(x=1,\lambda)$, which is less than or equal to $f_2(x_{\rm{th}},\lambda)$ of \Eq{eq: temp minimum}. If there is a crossing, instead, the minimum corresponds to the value of the functions at the intersection, which is $\frac{\lambda+1}{2\lambda}$. This concludes the proof.

\section{Proof of Proposition \ref{prop: B UB for ampd damp}}\label{app: Emax Amd Damp}
Because of \Eq{eq: SDP Emax} we need to show the relation $\Sigma(\mathcal N_\lambda^{(\rm{ad})}) = 2- \lambda$, where $\Sigma(\mathcal N_\lambda^{(\rm{ad})})$ has been defined in \Eq{eq:Sigma}. 
This is equivalent to showing that
\begin{eqnarray}
\min_{Y_{AB}\in\overrightarrow{\rm{SEP}}}\{
\Vert{\rm Tr}_B Y_{AB}\Vert_\infty:Y_{AB}-\pi_{\mathcal N_\lambda^{(\rm{ad})}}\geq 0
\}	
= 1-\frac12\lambda,\label{appeq: goal}
\end{eqnarray}
where $\Vert W\Vert_\infty := \max_{\left | \phi\right>} \left<\phi\right|W\left | \phi\right>$, and $\pi_{\mathcal N_\lambda^{(\rm{ad})}}$ is the normalized Choi state of the amplitude damping channel, which in basis $\{\ket{00},\ket{01},\ket{10},\ket{11}\}$ can be written as
\begin{eqnarray}
\pi_{\mathcal N_\lambda^{(\rm{ad})}}
:=\frac12\left(
\begin{array}{cccc}
1&0&0&\sqrt{1-\lambda}\\
0&0&0&0\\
0&0&\lambda&0\\
\sqrt{1-\lambda}&0&0&1-\lambda
\end{array}\right).
\end{eqnarray} 
Furthermore, as already observed in the main text, for qubit channels we can replace the cone of separable operators $\overrightarrow{\rm{SEP}}$ with that of PPT operators
\begin{eqnarray}
\overrightarrow{\rm{PPT}}&:=&\{V:V\geq 0\land V^{\rm{PT}}\geq 0\},
\end{eqnarray}
where the superscript $\rm{PT}$ represents partial transposition on the second qubit.

For the proof we exploit once again the symmetry of the channel under phase rotations, and we define a subset of $\overrightarrow{\rm{PPT}}$ as
\begin{align}
&\overrightarrow{\rm{PPT}}':=\{V:V\in\overrightarrow{\rm{PPT}} \land   U_\theta V U_\theta^\dagger =V \; \forall \theta\in \mathbb{R}\} \notag
\\
=&\left\{V:
V=\left(
\begin{array}{cccc}
\alpha&0&0&\xi e^{i\phi}\\
0&\gamma&0&0\\
0&0&\delta&0\\
\xi e^{-i\phi}&0&0&\beta
\end{array}\right)
\;\land\;
\alpha,\beta,\gamma,\delta,\xi\geq0
\;\land\;
\phi\in[0,2\pi]
\;\land\;
0\leq \xi\leq\min\{\sqrt{\alpha\beta},\sqrt{\gamma\delta}\}
\right\}, \label{appeq: long expr}
\end{align}
where $U_\theta$ is the unitary rotation defined in \Eq{eq: Utheta structure}. We now obtain a long sequence of equalities, which will be commented in the following. In particular, one has
\begin{align}
&\min_{Y_{AB}\in\overrightarrow{\rm{PPT}}}\left\{
\Vert{\rm Tr}_B Y_{AB}\Vert_\infty:Y_{AB}-\pi_{\mathcal N_\lambda^{(\rm{ad})}}\geq 0
\right\}
\nonumber\\
&\qquad=\min_{Y_{AB}\in\overrightarrow{\rm{PPT}}'}\{
||{\rm Tr}_B Y_{AB}||_\infty:Y_{AB}-\pi_{\mathcal N_\lambda^{(\rm{ad})}}\geq 0
\}
\nonumber\\
&\qquad=\min\left\{
\Vert{\rm Tr}_B V\Vert_\infty:V-\pi_{\mathcal N_\lambda^{(\rm{ad})}}\geq 0
\;\land\;
V\in\overrightarrow{\rm{PPT}}'
\right\}
\nonumber\\
&\qquad=\min\bigg\{\max\{\alpha+\gamma,\delta+\beta\}
:
\left(
\begin{array}{cc}
\alpha&\xi e^{i\phi}\\
\xi e^{-i\phi}&\beta
\end{array}\right)
-
\frac12\left(
\begin{array}{cc}
1&\sqrt{1-\lambda}\\
\sqrt{1-\lambda}&1-\lambda
\end{array}\right)
\geq 0
\nonumber\\&\qquad\qquad\qquad
\;\land\;
\delta-\frac12\lambda\geq0
\;\land\;
\alpha,\beta,\gamma,\delta\geq0
\;\land\;
\phi\in[0,2\pi]
\;\land\;
0\leq \xi\leq\min\{\sqrt{\alpha\beta},\sqrt{\gamma\delta}\}
\bigg\}
\nonumber\\
&\qquad=\min\bigg\{\max\{\alpha+\gamma,\delta+\beta\}
:
\alpha+\beta\geq 1-\frac12\lambda
\;\land\;
\left(\alpha-\frac12\right)\left[\beta-\frac12(1-\lambda)\right]\geq
\left|\xi e^{i\phi}-\frac12\sqrt{1-\lambda}\right|^2
\nonumber\\&\qquad\qquad\qquad
\;\land\;
\delta-\frac12\lambda\geq0
\;\land\;
\alpha,\beta,\gamma,\delta\geq0
\;\land\;
\phi\in[0,2\pi]
\;\land\;
0\leq \xi\leq\min\{\sqrt{\alpha\beta},\sqrt{\gamma\delta}\}
\bigg\}
\nonumber\\
&\qquad=\min\bigg\{\max\{\alpha+\gamma,\delta+\beta\}
:
\alpha+\beta\geq 1-\frac12\lambda
\;\land\;
\left(\alpha-\frac12\right)\left[\beta-\frac12(1-\lambda)\right]\geq
\left(\xi -\frac12\sqrt{1-\lambda}\right)^2
\nonumber\\&\qquad\qquad\qquad
\;\land\;
\delta-\frac12\lambda\geq0
\;\land\;
\alpha,\beta,\gamma,\delta\geq0
\;\land\;
0\leq \xi\leq\min\{\sqrt{\alpha\beta},\sqrt{\gamma\delta}\}
\bigg\}
\nonumber\\
&\qquad=\min\bigg\{\max\{\alpha+\gamma,\delta+\beta\}
:
\alpha+\beta\geq 1-\frac12\lambda
\;\land\;
\delta-\frac12\lambda\geq0
\;\land\;
\alpha,\beta,\gamma,\delta\geq0
\nonumber\\&\qquad\qquad\qquad
\;\land\;
\left(\alpha-\frac12\right)\left[\beta-\frac12(1-\lambda)\right]\geq
\left(\min\{\min\{\sqrt{\alpha\beta},\sqrt{\gamma\delta}\} -\frac12\sqrt{1-\lambda},0\}\right)^2
\bigg\}
\nonumber\\
&\qquad=\min\bigg\{\max\{\alpha+\gamma,\delta+\beta\}
: \delta-\frac12\lambda\geq0
\;\land\;
\alpha\geq \frac 12
\;\land\;
\beta\geq \frac12(1-\lambda)
\;\land\;
\gamma,\delta\geq0
\nonumber\\&\qquad\qquad\qquad
\;\land\;
\left(\alpha-\frac12\right)\left[\beta-\frac12(1-\lambda)\right]\geq
\left(\min\{\min\{\sqrt{\alpha\beta},\sqrt{\gamma\delta}\} -\frac12\sqrt{1-\lambda},0\}\right)^2
\bigg\}
\nonumber\\
&\qquad=
\min\{A,B\},\end{align}
where
\begin{eqnarray}
A&:=&\min\Big\{
\frac12\left(\alpha+\beta+\sqrt{(\alpha-\beta)^2+4x^2}\right)
:
\alpha+\frac{2x^2}\lambda\geq \beta+\frac \lambda 2
\;\land\;
\alpha\geq \frac 12
\;\land\;
\beta\geq \frac12(1-\lambda)
\;\land\;
x\geq0
\nonumber\\&&{}
\;\land\;
\left(\alpha-\frac12\right)\left[\beta-\frac12(1-\lambda)\right]\geq
\left(\min\{\min\{\sqrt{\alpha\beta},x\} -\frac12\sqrt{1-\lambda},0\}\right)^2
\bigg\}, \label{appeq: Adef}
\\
B&:=&\min\Big\{\beta+\frac12\lambda
:
\alpha+\frac{2x^2}\lambda\leq \beta+\frac \lambda 2
\;\land\;
\alpha\geq \frac 12
\;\land\;
\beta\geq \frac12(1-\lambda)
\;\land\;
x\geq0
\nonumber\\&&{}
\;\land\;
\left(\alpha-\frac12\right)\left[\beta-\frac12(1-\lambda)\right]\geq
\left(\min\{\min\{\sqrt{\alpha\beta},x\} -\frac12\sqrt{1-\lambda},0\}\right)^2
\bigg\}. \label{appeq: B def}
\end{eqnarray}
The first equality comes from the following two observations 
\begin{eqnarray}
V\in \overrightarrow{\rm{PPT}}&\Rightarrow &V'\in \overrightarrow{\rm{PPT}}',
\\
V-\pi_{\mathcal N_\lambda^{(\rm{ad})}}\geq 0
&\Rightarrow &
V'-\pi_{\mathcal N_\lambda^{(\rm{ad})}},
\end{eqnarray}
and from the inequality $\Vert{\rm Tr}_B V\Vert_\infty \geq \Vert{\rm Tr}_B V'\Vert_\infty$,
where $V, V'$ are generic matrices of the form 
\begin{equation}
V = \left(
\begin{array}{cccc}
a_{11}&a_{12}&a_{13}&a_{14}\\
a_{21}&a_{22}&a_{23}&a_{24}\\
a_{31}&a_{32}&a_{33}&a_{34}\\
a_{41}&a_{42}&a_{43}&a_{44}
\end{array}\right),\qquad \qquad
V' = \left(
\begin{array}{cccc}
a_{11}&0&0&a_{14}\\
0&a_{22}&0&0\\
0&0&a_{33}&0\\
a_{41}&0&0&a_{44}
\end{array}\right).
\end{equation}
The second equality is just a rearrangement of the previous expression, whereas in the third equality we exploit \Eq{appeq: long expr}. The fourth equality can be obtained by expanding the matrix inequality previously found, and in the fifth equality we used the following relation 
\begin{equation}
\min_{\phi\in \mathbb{R}}\left|\xi e^{i\phi}-\frac12\sqrt{1-\lambda}\right|^2=\left(\xi -\frac12\sqrt{1-\lambda}\right)^2,
\end{equation}
which holds for $\xi\geq 0$.
In the sixth and seventh equalities we used respectively
\begin{equation}
x,y\geq0 \Rightarrow \min_{0\leq \xi\leq x}(\xi-y)^2=\min\{x-y,0\}^2, 
\end{equation}
and 
\begin{equation}
\alpha+\beta\geq x+y\land(\alpha-x)(\beta-y)\geq 0\;\Leftrightarrow\;\alpha\geq x \land \beta\geq y.
\end{equation}
Finally, in order to obtain the last equality we observed that 
\begin{eqnarray}
\min_{\delta\geq \frac12\lambda}\max\left\{\alpha+\frac {x^2}\delta,\beta+\delta\right\}
&=&\left\{
\begin{array}{cl}
\frac12(\alpha+\beta+\sqrt{(\alpha-\beta)^2+4x^2}),&\makebox{for $\alpha+\frac{2x^2}\lambda\geq \beta+\frac \lambda 2
	$},\\
\beta+\frac12\lambda,&\makebox{for $\alpha+\frac{2x^2}\lambda\leq \beta+\frac \lambda 2
	$}.
\end{array}
\right.
\end{eqnarray}

From this analysis it follows that \Eq{appeq: goal} is proven if we can show that $A=1-\frac12\lambda$ and $B\geq 1-\frac12 \lambda$. This is what we do in the following.

\subsection{Proof of $A=1-\frac12\lambda$}
Note that the choices $\alpha=\frac12$, $\beta=\frac12(1-\lambda)$, and $x=\frac 12\sqrt{1-\lambda}$ satisfy the conditions appearing in \Eq{appeq: Adef}, providing
\begin{equation}
\frac12(\alpha+\beta+\sqrt{(\alpha-\beta)^2+4x^2}) = 1-\frac12\lambda,
\end{equation}
so that $A\leq 1-\frac12 \lambda$.

In order to derive the converse inequality, we first rewrite $A$ as
\begin{eqnarray}
A&=&\min\Bigg\{
\frac12\left[ a+b+1-\frac12\lambda+\sqrt{\Big(a-b+\frac12\lambda\Big)^2+4x^2}\right]
:
a+\frac{2x^2}\lambda\geq b
\;\land\;
a,b,x\geq0
\nonumber\\&&{}
\land\; ab\geq
\left(\min\left\{\min\left\{\sqrt{\left(a+\frac12\right)\left[b+\frac12(1-\lambda)\right]},x\right\} -\frac12\sqrt{1-\lambda},0\right\}\right)^2
\Bigg\}
\label{appeq: line 1}\\
&\geq&\min\Bigg\{
\frac12\left[a+b+1-\frac12\lambda+\sqrt{\Big(a-b+\frac12\lambda\Big)^2+4x^2}\right]
: a,b,x\geq0
\nonumber\\&&{}
\land \; ab\geq
\left(\min\left\{\min\left\{\sqrt{\left(a+\frac12\right)\left[b+\frac12(1-\lambda)\right]},x\right\} -\frac12\sqrt{1-\lambda},0\right\}\right)^2
\Bigg\}
\label{appeq: line 2}\\
&=&\min\{A_1,A_2\}, \label{appeq: line 3}
\end{eqnarray}
where
\begin{eqnarray}
A_1&:=&
\min\Bigg\{
\frac12\left[a+b+1-\frac12\lambda+\sqrt{\Big(a-b+\frac12\lambda\Big)^2+4x^2}\right]
: a,b,x\geq0
\nonumber\\&&{}
\land\; \min\bigg\{\sqrt{\Big(a+\frac12\Big)\Big[b+\frac12(1-\lambda)\Big]},x\bigg\} \geq \frac12\sqrt{1-\lambda}
\Bigg\},
\\
A_2&:=&
\min\Bigg\{
\frac12\left[a+b+1-\frac12\lambda+\sqrt{\Big(a-b+\frac12\lambda\Big)^2+4x^2}\right]
: a,b,x\geq0
\nonumber\\&&{}
\land\; ab\geq
\left(\min\left\{\sqrt{\Big(a+\frac12\Big)\left[b+\frac12(1-\lambda)\right]},x\right\} -\frac12\sqrt{1-\lambda}\right)^2
\nonumber\\&&{}
\;\land\;
\min\left\{\sqrt{\Big(a+\frac12\Big)\left[b+\frac12(1-\lambda)\right]},x\right\} \leq \frac12\sqrt{1-\lambda}
\Bigg\}.
\end{eqnarray}
The equality in \Eq{appeq: line 1} follows from the definition of $A$ in \Eq{appeq: Adef}, by substituting
$\alpha\rightarrow a+\frac12$ and 
$\beta\rightarrow b+\frac12(1-\lambda)$, and in order to obtain the following inequality we drop a condition on the parameters $a,b,x$.
Then, the equality in \Eq{appeq: line 3} can be proven by dividing the parameter region into two sub-regions: one such that 
$\min\Big\{\sqrt{(a+\frac12)(b+\frac12(1-\lambda))},x\Big\} \geq \frac12\sqrt{1-\lambda}$, and one such that the converse inequality holds.

As next step, we show that both $A_1$ and $A_2$ are larger than $1-\frac12\lambda$. In particular, we can explicitly evaluate $A_1$ as
\begin{eqnarray}
A_1&=&
\min\Bigg\{
\frac12\left(a+b+1-\frac12\lambda+\sqrt{\Big(a-b+\frac12\lambda\Big)^2+(1-\lambda)}\right)
:
\nonumber\\&&{}
\sqrt{\Big(a+\frac12\Big)\left[b+\frac12(1-\lambda)\right]} \geq \frac12\sqrt{1-\lambda}
\;\land\;
a,b\geq0
\Bigg\}
\nonumber\\
&=&
\min\Bigg\{
\frac12\left(a+b+1-\frac12\lambda+\sqrt{\Big(a-b+\frac12\lambda\Big)^2+(1-\lambda)}\right)
:
\nonumber\\&&{}
\sqrt{\Big(a+\frac12\Big)\left[b+\frac12(1-\lambda)\right]} \geq \frac12\sqrt{1-\lambda}
\;\land\;a=b=0
\Bigg\} = 1-\frac12\lambda.
\end{eqnarray}
The two equalities can be respectively shown by noticing that the quantity being minimised is a monotonically increasing function of $x\geq 0$ and of $a,b\geq 0$. In order to show that $A_2 \geq 1-\frac12\lambda$, it is convenient to write
\begin{eqnarray}
A_2&=&\min\{A_{3},A_{4}\},
\end{eqnarray}
where we divide the parameter region into two sub-regions, depending on the ordering between 
$\sqrt{(a+\frac12)\left[b+\frac12(1-\lambda)\right]}$ and $x$, that is:
\begin{eqnarray}
A_{3}&:=&
\min\Bigg\{
\frac12\left[ a+b+1-\frac12\lambda+\sqrt{\Big(a-b+\frac12\lambda\Big)^2+4x^2}\right]
: a,b,x\geq0
\nonumber\\&&{}
\land\; ab\geq
\left(\sqrt{\Big(a+\frac12\Big)\left[b+\frac12(1-\lambda)\right]} -\frac12\sqrt{1-\lambda}\right)^2
\;\land\;
\sqrt{\Big(a+\frac12\Big)\left[b+\frac12(1-\lambda)\right]} \leq \frac12\sqrt{1-\lambda}
\nonumber\\&&{}
\;\land\;
\sqrt{\Big(a+\frac12\Big)\left[b+\frac12(1-\lambda)\right]}\leq x
\Bigg\}
\nonumber\\
&=&
\min\bigg\{
\frac12\left(1-\frac12\lambda+\sqrt{\frac{\lambda^2}4+4x^2}\right)
:
\frac12\sqrt{1-\lambda}\leq x
\;\land\;
x\geq0
\bigg\}
= 1-\frac12\lambda,
\end{eqnarray}
and
\begin{eqnarray}
A_{4}&:=&
\min\Bigg\{
\frac12\left[ a+b+1-\frac12\lambda+\sqrt{\Big(a-b+\frac12\lambda\Big)^2+4x^2}\right]
: a,b,x\geq0
\nonumber\\&&{}
\land\; ab\geq
\left(x -\frac12\sqrt{1-\lambda}\right)^2
\;\land\;
x \leq \frac12\sqrt{1-\lambda}
\;\land\;
\sqrt{\Big(a+\frac12\Big)\left[b+\frac12(1-\lambda)\right]}\geq x
\Bigg\}
\nonumber\\
&\geq&
\min\Bigg\{
\frac12\left[ a+b+1-\frac12\lambda+\sqrt{\Big(a-b+\frac12\lambda\Big)^2+4x^2}\right]
: a,b,x\geq0
\nonumber\\&&{}
\land\; ab\geq
\left(\frac12\sqrt{1-\lambda}-x\right)^2
\;\land\;
x \leq \frac12\sqrt{1-\lambda}
\Bigg\}
\nonumber\\
&=&
\min\Bigg\{
\frac12\left[ a+b+1-\frac12\lambda+\sqrt{\Big(a-b+\frac12\lambda\Big)^2+4x^2}\right]
: a,b,x\geq0\nonumber\\&&{}
\land\; \frac12\sqrt{1-\lambda}-\sqrt{ab}\leq
x \leq \frac12\sqrt{1-\lambda}
\Bigg\}
\nonumber\\
&=&
\min\{A_{5},A_{6}\},
\end{eqnarray}
where we further expanded $A_4$ in terms of $A_5$ and $A_6$, depending on the ordering between $\frac12\sqrt{1-\lambda}$ and $\sqrt{ab}$:
\begin{eqnarray}
A_5&:=&
\min\Bigg\{
\frac12\left[ a+b+1-\frac12\lambda+\sqrt{\Big(a-b+\frac12\lambda\Big)^2+4x^2}\right]
: a,b,x\geq0 \nonumber\\&&
\land\; \frac12\sqrt{1-\lambda}-\sqrt{ab}\leq
x \leq \frac12\sqrt{1-\lambda}
\;\land\;
\frac12\sqrt{1-\lambda}\geq\sqrt{ab}
\Bigg\}
\nonumber\\
&=&
\min\Bigg\{
\frac12\left[a+b+1-\frac12\lambda+\sqrt{\Big(a-b+\frac12\lambda\Big)^2+(\sqrt{1-\lambda}-2\sqrt{ab})^2}\right]
:\nonumber\\&&
\quad \frac12\sqrt{1-\lambda}\geq\sqrt{ab}
\;\land\;
a,b\geq0
\Bigg\}
\nonumber\\
&\geq&
\min\Bigg\{
\frac12\left[a+b+1-\frac12\lambda+\sqrt{\Big(a-b+\frac12\lambda\Big)^2+(\sqrt{1-\lambda}-2\sqrt{ab})^2}\right]
:
a,b\geq0
\Bigg\}
\nonumber\\
&=&
\min\Bigg\{
\frac12\left[x+1-\frac12\lambda+\sqrt{\Big(1-\frac12\lambda\Big)^2+x^2+x\left(y  \lambda-2\sqrt{1-\lambda}\sqrt{1-y^2}\right)
}\right]
: \nonumber\\&&
\quad 1\geq y\geq -1\land x\geq0
\Bigg\}
\nonumber\\
&=&
\min\left\{
\frac12\left[x+1-\frac12\lambda+\Big|1-\frac12\lambda-x\Big|\right]
:
x\geq0
\right\}=1-\frac12\lambda,\label{appeq: A5}\\
&&\nonumber\\
A_6&:=&
\min\Bigg\{
\frac12\left[ a+b+1-\frac12\lambda+\sqrt{\Big(a-b+\frac12\lambda\Big)^2+4x^2}\right]
: a,b,x\geq0 \nonumber\\&&
\;\land\;\frac12\sqrt{1-\lambda}-\sqrt{ab}\leq
x \leq \frac12\sqrt{1-\lambda}
\;\land\;
\frac12\sqrt{1-\lambda}\leq\sqrt{ab}
\Bigg\}
\nonumber\\
&=&
\min\bigg\{
\frac12\left[a+b+1-\frac12\lambda+\Big|a-b+\frac12\lambda\Big|\right]
:
\frac12\sqrt{1-\lambda}\leq\sqrt{ab}
\;\land\;
a,b\geq0
\bigg\}
\nonumber\\
&=&
\min\bigg\{
\frac12\left[\sqrt{4x^2+y^2}+1-\frac12\lambda+\Big|y+\frac12\lambda\Big|\right]
:
\frac12\sqrt{1-\lambda}\leq x
\;\land\;
x\geq0
\bigg\}
\nonumber\\
&=&
\min\bigg\{
\frac12\left[\sqrt{1-\lambda+y^2}+1-\frac12\lambda+\Big|y+\frac12\lambda\Big|\right]
\bigg\}=1-\frac12\lambda. \label{appeq: A6}
\end{eqnarray}
In order to manipulate the expression of $A_3$, we first used the fact that 
\begin{equation}
\sqrt{\Big(a+\frac12\Big)\left[b+\frac12(1-\lambda)\right]}\leq \frac12\sqrt{1-\lambda}\land a,b\geq 0\Leftrightarrow a=0\land b=0,
\end{equation}
and then we exploited the monotonicity of $\frac12\Big(1-\frac12\lambda+\sqrt{\frac{\lambda^2}4+4x^2}\Big)$ in $x$ for $x\geq 0$.
The inequalities appearing in the manipulations of $A_4$ and $A_5$ are obtained by dropping a condition on $a,b,x$ which restricts the minimisation region.
In the first equalities written for $A_5$ and $A_6$ we used the monotonicity in $x$ of the function 
\begin{equation}
\frac12\left[a+b+1-\frac12\lambda+\sqrt{\left(a-b+\frac12\lambda\right)^2+4x^2}\right],
\end{equation}
which is minimised for $x=0$.
In the third relation appearing in the manipulation of $A_5$ we changed variables as 
$a\rightarrow \frac12(x+xy)$ and $b\rightarrow \frac12(x-xy)$, whereas in the second relation appearing in the manipulation of $A_6$ we parametrized $a,b$ as $a\rightarrow \frac12(y+\sqrt{4x^2+y^2})$ and $b\rightarrow \frac12(-y+\sqrt{4x^2+y^2})$.
Finally, the last equalities leading to the evaluation of $A_5$ and $A_6$
in \Eq{appeq: A5} and \Eq{appeq: A6} are respectively due to
\begin{equation}
\min_{-1\leq y\leq 1} \left[y  \lambda-2\sqrt{1-\lambda}\sqrt{1-y^2} \right] =
-(2-\lambda),
\end{equation}
and
\begin{equation}
\min_y\left[\sqrt{1-\lambda+y^2}+\left|y+\frac12\lambda\right| \right] = 1-\frac12\lambda.
\end{equation}

Overall, we have been able to show that $A_1 = 1-\frac12\lambda$, and that $A_2$ can be written as the minimum among quantities larger than or equal to $1-\frac12\lambda$. Therefore, from \Eq{appeq: line 3} it follows that $A=1-\frac12\lambda$, as desired.

\subsection{Proof of  $B\geq1-\frac12\lambda$}
We start by writing the values of $\alpha$ and $\beta$ appearing in the definition of $B$ in \Eq{appeq: B def} as
$\alpha\rightarrow a+\frac12$ and
$\beta\rightarrow b+\frac12(1-\lambda)$. Then, we divide the parameter region into two sub-regions depending on the ordering between 
$\min\Big\{\sqrt{(a+\frac12)\left[b+\frac12(1-\lambda)\right]},x\Big\}$ and $\frac12\sqrt{1-\lambda}$. This leaves us with

\begin{eqnarray}
B&=&\min\Bigg\{b+\frac12
:
ab\geq
\left(\min\left\{\min\left\{\sqrt{\Big(a+\frac12\Big)\left[b+\frac12(1-\lambda)\right]},x\right\} -\frac12\sqrt{1-\lambda},0\right\}\right)^2
\nonumber\\&&{}
\;\land\;
a+\frac{2x^2}\lambda\leq b
\;\land\;
a,b,x\geq0
\Bigg\}
\nonumber\\
&=&\min\{B_1,B_2\},
\end{eqnarray}
where
\begin{eqnarray}
B_1&:=&
\min\bigg\{b+\frac12
:
ab\geq0
\;\land\;
\min\left\{\sqrt{\Big(a+\frac12\Big)\left[b+\frac12(1-\lambda)\right]},x\right\} \geq\frac12\sqrt{1-\lambda}
\nonumber\\&&{}
\;\land\;
a+\frac{2x^2}\lambda\leq b
\;\land\;
a,b,x\geq0
\bigg\}
\nonumber\\
&=&
\min\bigg\{b+\frac12
:
\sqrt{\Big(a+\frac12\Big)\left[b+\frac12(1-\lambda)\right]} \geq\frac12\sqrt{1-\lambda}
\nonumber\\&&{}
\;\land\;
x \geq\frac12\sqrt{1-\lambda}
\;\land\;
a+\frac{2x^2}\lambda\leq b
\;\land\;
a,b,x\geq0
\bigg\}
\nonumber\\
\nonumber\\
&\geq&
\min\left\{b+\frac12
:
x \geq\frac12\sqrt{1-\lambda}
\;\land\;
a+\frac{2x^2}\lambda\leq b
\;\land\;
a,b,x\geq0
\right\}
\nonumber\\
&=&
\min\left\{b+\frac12
:
a+\frac{1-\lambda}\lambda\leq b
\;\land\;
a\geq0
\right\}
\nonumber\\
&=&
\min\left\{b+\frac12
:
\frac{1-\lambda}\lambda\leq b
\right\} = 
\frac1\lambda (1-\frac12\lambda)\geq 1-\frac12\lambda,
\end{eqnarray}
and 
\begin{eqnarray}
B_2
&:=&
\min\bigg\{b+\frac12
:
ab\geq
\left(\min\left\{\sqrt{\Big(a+\frac12\Big)\left[b+\frac12(1-\lambda)\right]},x\right\} -\frac12\sqrt{1-\lambda}\right)^2
\nonumber\\&&{}
\;\land\;
\min\left\{\sqrt{\Big(a+\frac12\Big)\left[b+\frac12(1-\lambda)\right]},x\right\} \leq\frac12\sqrt{1-\lambda}
\;\land\;
a+\frac{2x^2}\lambda\leq b
\;\land\;
a,b,x\geq0
\bigg\}
\nonumber\\
&=&
\min\bigg\{b+\frac12
:
\sqrt{ab}\geq
\frac12\sqrt{1-\lambda}-\min\left\{\sqrt{\Big(a+\frac12\Big)\left[b+\frac12(1-\lambda)\right]},x\right\}
\nonumber\\&&{}
\;\land\;
\min\left\{\sqrt{\Big(a+\frac12\Big)\left[b+\frac12(1-\lambda)\right]},x\right\} \leq\frac12\sqrt{1-\lambda}
\;\land\;
a+\frac{2x^2}\lambda\leq b
\;\land\;
a,b,x\geq0
\bigg\}
\nonumber\\
&\geq&
\min\bigg\{b+\frac12
:
\sqrt{ab}\geq
\frac12\sqrt{1-\lambda}-\min\left\{\sqrt{\Big(a+\frac12\Big)\left[b+\frac12(1-\lambda)\right]},x\right\}
\nonumber\\&&{}
\;\land\;
a+\frac{2x^2}\lambda\leq b
\;\land\;
a,b,x\geq0
\bigg\}
\nonumber\\
&=&
\min\bigg\{b+\frac12
:
\sqrt{ab}\geq
\frac12\sqrt{1-\lambda}-x 
\;\land\;
\sqrt{ab}\geq
\frac12\sqrt{1-\lambda}-\sqrt{\Big(a+\frac12\Big)\left[b+\frac12(1-\lambda)\right]} 
\nonumber\\&&{}
\;\land\;
a+\frac{2x^2}\lambda\leq b
\;\land\;
a,b,x\geq0
\bigg\}
\nonumber\\
&\geq&
\min\left\{b+\frac12
:
\sqrt{ab}\geq
\frac12\sqrt{1-\lambda}-x 
\;\land\;
a+\frac{2x^2}\lambda\leq b
\;\land\;
a,b,x\geq0
\right\}
\nonumber\\
&=&\min\{B_{3},B_{4}\}.
\end{eqnarray}
The inequalities appearing in the above manipulations are obtained by dropping conditions on $a,b,x$ which restrict the minimisation region. Moreover, $B_3$ and $B_4$ are obtained by splitting the parameter region into two sub-regions, defined according to the ordering between $\sqrt{ab}$ and $\frac12\sqrt{1-\lambda}$. More precisely, we can write
\begin{eqnarray}
B_{3}
&:=&
\min\left\{b+\frac12
:
x\geq
\frac12\sqrt{1-\lambda}-\sqrt{ab}
\;\land\;
\sqrt{ab}\geq
\frac12\sqrt{1-\lambda}
\;\land\;
a+\frac{2x^2}\lambda\leq b
\;\land\;
a,b,x\geq0
\right\}
\nonumber\\
&=&
\min\left\{b+\frac12
:
\sqrt{ab}\geq
\frac12\sqrt{1-\lambda}
\;\land\;
a\leq b
\;\land\;
a,b\geq0
\right\}
\nonumber\\
&=&
\min\left\{b+\frac12
:
b\geq
\frac12\sqrt{1-\lambda}
\;\land\;
b\geq0
\right\}=\frac12\sqrt{1-\lambda}+\frac12\geq1-\frac12\lambda,
\end{eqnarray}
and 
\begin{eqnarray}
B_{4}
&:=&
\min\left\{b+\frac12
:
x\geq
\frac12\sqrt{1-\lambda}-\sqrt{ab}
\;\land\;
\sqrt{ab}\leq
\frac12\sqrt{1-\lambda}
\;\land\;
a+\frac{2x^2}\lambda\leq b
\;\land\;
a,b,x\geq0
\right\}
\nonumber\\
&=&
\min\left\{b+\frac12
:
\sqrt{ab}\leq
\frac12\sqrt{1-\lambda}
\;\land\;
a+\frac{2\left(\frac12\sqrt{1-\lambda}-\sqrt{ab}\right)^2}\lambda\leq b
\;\land\;
a,b\geq0
\right\}
\nonumber\\
&=&
\min\left\{b+\frac12
:
\sqrt{ab}\leq
\frac12\sqrt{1-\lambda}
\;\land\;
\frac12\sqrt{1-\lambda}-\sqrt{ab}\leq \sqrt{\frac\lambda2}\sqrt{b-a}
\;\land\;
b\geq a
\;\land\;
a\geq0
\right\}
\nonumber\\
&=&
\min\left\{
\frac12\left(1+y^2+\sqrt{4x^2+y^4}\right)
:
x\leq
\frac12\sqrt{1-\lambda}
\;\land\;
\frac12\sqrt{1-\lambda}-x\leq \sqrt{\frac\lambda2}y
\;\land\;
x,y\geq0
\right\}
\nonumber\\
&=&
\min\left\{\frac12\left(1+y^2+\sqrt{4x^2+y^4}\right)
:
x\leq
\frac12\sqrt{1-\lambda}
\;\land\;
\frac12\sqrt{1-\lambda}-x= \sqrt{\frac\lambda2}y
\;\land\;
x\geq0
\right\}
\nonumber\\
&=&
\min\left\{\frac12\left[1+y^2+\sqrt{(\sqrt{1-\lambda}- \sqrt{2\lambda }y)^2+y^4}\right]
:
0\leq
y
\leq
\sqrt{\frac{1-\lambda}{2\lambda}}
\right\}
\nonumber\\
&\geq&
\min_{y\in \mathbb R}\left\{\frac12\left[1+y^2+\sqrt{(\sqrt{1-\lambda}- \sqrt{2\lambda }y)^2+y^4}\right]
\right\}
\nonumber\\
&=&
\min_{z\in \mathbb R}\left\{\frac12+\frac14(1-\lambda)\lambda\left[ z^2+\sqrt{\frac{4(1-\lambda z)^2}{(1-\lambda)\lambda^2}+z^4}\right]
\right\}
\nonumber\\
&=&
\frac12+\frac14(1-\lambda)\lambda\left[ 1+\sqrt{\frac{4(1-\lambda )^2}{(1-\lambda)\lambda^2}+1}\right]=1-\frac12\lambda.
\end{eqnarray}
In the third equality of the manipulations performed on $B_4$ we changed variables as 
$b\rightarrow \frac12(y^2+\sqrt{4x^2+y^4})$ and
$a\rightarrow \frac12(-y^2+\sqrt{4x^2+y^4})$, whereas in the fourth equality we used the fact that 
$\frac12(1+y^2+\sqrt{4x^2+y^4})$ is a monotonic function of $y$ when $y\geq 0$.
The seventh and eighth relations appearing in the manipulation of $B_4$, instead, are respectively obtained by changing variable as $y\rightarrow \sqrt{ \frac{(1-\lambda)\lambda}{2}} z$ and by exploiting the fact that $ z^2+\sqrt{\frac{4(1-\lambda z)^2}{(1-\lambda)\lambda^2}+z^4}$ is minimised at $z=1$ for $0<\lambda<1$. 

Overall, this shows that $B\geq 1-\frac 12 \lambda$ and the proof is concluded.

\end{document}